\documentclass[twocolumn,english]{svjour3}
\usepackage[T1]{fontenc}
\usepackage[latin9]{inputenc}
\usepackage[active]{srcltx}
\usepackage{array}
\usepackage{float}
\usepackage{textcomp}
\usepackage{url}
\usepackage{multirow}
\usepackage{amsmath}
\usepackage{amssymb}
\usepackage{graphicx}

\makeatletter

\providecommand{\tabularnewline}{\\}
\newcommand{\lyxdot}{.}

\floatstyle{ruled}

\newfloat{algorithm}{tbp}{loa}[chapter]
\providecommand{\algorithmname}{Algorithm}
\floatname{algorithm}{\protect\algorithmname}

\usepackage{cite}
\usepackage{doi}

\makeatother

\usepackage{babel}

\begin{document}

\title{Realtime Active Sound Source Localization for Unmanned Ground Robots
Using a Self-Rotational Bi-Microphone Array}

\author{Deepak Gala, Nathan Lindsay, Liang Sun}

\institute{Deepak Gala, drgala@nmsu.edu, New Mexico State University, NM, USA;
Nathan Lindsay, nl22@nmsu.edu, New Mexico State University, NM, USA;
Liang Sun, lsun@nmsu.edu, New Mexico State University, NM, USA}

\maketitle
This work presents a novel technique that performs both orientation
and distance localization of a sound source in a three-dimensional
(3D) space using only the interaural time difference (ITD) cue, generated
by a newly-developed self-rotational bi-microphone robotic platform.
The system dynamics is established in the spherical coordinate frame
using a state-space model. The observability analysis of the state-space
model shows that the system is unobservable when the sound source
is placed with elevation angles of $90$ and $0$ degree. The proposed
method utilizes the difference between the azimuth estimates resulting
from respectively the 3D and the two-dimensional models to check the
zero-degree-elevation condition and further estimates the elevation
angle using a polynomial curve fitting approach. Also, the proposed
method is capable of detecting a $90$-degree elevation by extracting
the zero-ITD signal 'buried' in noise. Additionally, a distance localization
is performed by first rotating the microphone array to face toward
the sound source and then shifting the microphone perpendicular to
the source-robot vector by a predefined distance of a fixed number
of steps. The integrated rotational and translational motions of the
microphone array provide a complete orientation and distance localization
using only the ITD cue. A novel robotic platform using a self-rotational
bi-microphone array was also developed for unmanned ground robots
performing sound source localization. The proposed technique was first
tested in simulation and was then verified on the newly-developed
robotic platform. Experimental data collected by the microphones installed
on a KEMAR dummy head were also used to test the proposed technique.
All results show the effectiveness of the proposed technique. 

\section{INTRODUCTION}

The localization problem in the robotic field has been recognized
as the most fundamental problem to make robots truly autonomous~\cite{Borenstein1996}.
Localization techniques are of great importance for autonomous unmanned
systems to identify their own locations (i.e., self-localization)
and situational awareness (e.g., locations of surrounding objects),
especially in an unknown environment. Mainstream technology for localization
is based on computer vision, supported by visual sensors (e.g., cameras),
which, however, are subject to lighting and line-of-sight conditions
and rely on computationally demanding image-processing algorithms.
An acoustic sensor (e.g., a microphone), as a complementary component
in a robotic sensing system, does not require a line of sight and
is able to work under varying light (or completely dark) conditions
in an omnidirectional manner. Thanks to the advancement of microelectromechanical
technology, microphones become inexpensive and do not require significant
power to operate. 

Sound-source localization (SSL) techniques have been developed that
identify the location of sound sources (e.g., speech and music) in
terms of directions and distances. SSL techniques have been widely
used in civilian applications, such as intelligent video conferencing~\cite{huang2000passive,wang1997voice},
environmental monitoring~\cite{tiete2014soundcompass}, human-robot
interaction (HRI) for humanoid robotics~\cite{hornstein2006sound},
and robot motion planning~\cite{NguyenColasVincentEtAl2017}, as
well as military applications, such as passive sonar for submarine
detections, surveillance systems that locate hostile tanks, artillery,
incoming missiles~\cite{kaushik2005review}, aircraft~\cite{blumrich2000medium},
and UAVs~\cite{brandes2007sound}. SSL techniques have great potential
by itself to enhance the sensing capability of autonomous unmanned
systems as well as working together with vision-based localization
techniques. 

SSL has been achieved by using microphone arrays with more than two
microphones~\cite{TamaiSasakiKagamiEtAl2005,TamaiKagamiAmemiyaEtAl2004,SturimBrandsteinSilverman1997,ValinMichaudRouatEtAl2016,omologo1996acoustic}.
The accuracy of the localization techniques based on microphone arrays
is dictated by their physical sizes~\cite{brandstein2013microphone,benesty2008microphone,Zietlow2017}.
Microphone arrays are usually designed using particular (e.g., linear
or circular) structures, which result in their relatively large sizes
and sophisticated control components for operation. Therefore, it
becomes difficult to use them on small robots nor large systems due
to the complexity of mounting and maneuvering. 

In the past decade, research has been carried out for robots to have
auditory behaviors (e.g. getting attention to an event, locating a
sound source in potentially dangerous situations, and locating and
paying attention to a speaker) by mimicking human auditory systems.
Humans perform sound localization with their two ears using integrated
three types of cues, i.e., the interaural level difference (ILD),
the interaural time difference (ITD), and the spectral information~\cite{goldstein2016sensation,middlebrooks1991sound}.
ILD and ITD cues are usually used respectively to identify the horizontal
location (i.e., azimuth angle) of a sound source with higher and lower
frequencies. Spectral cues are usually used to identify the vertical
location (i.e., elevation angle) of a sound source with higher frequencies.
Additionally, acoustic landmarks aid towards bettering the SSL by
humans\cite{Zhong2015}. 

To mimic human acoustic systems, researchers have developed sound
source localization techniques using only two microphones. All three
types of cues have been used by Rodemann et al.~\cite{RodemannInceJoublinEtAl2008}
in a binaural approach of estimating the azimuth angle of a sound
source, while the authors also stated that reliable elevation estimation
would need a third microphone. Spectral cues were used by the head-related-transfer-function
(HRTF) that was applied to identify both the azimuth and elevation
angles of a sound source for binaural sensor platforms~\cite{Keyrouz2014,gill2000auditory,hornstein2006sound,KeyrouzDiepold2006}.
The ITD cues have also been used in binaural sound source localization~\cite{chen2006time},
where the problem of cone of confusion~\cite{Wallach1939} has been
overcome by incorporating head movements, which also enable both azimuth
and elevation estimation~\cite{Wallach1939,perrett1997effect}. Lu
et al.~\cite{lu2011motion} used a particle filter for binaural tracking
of a mobile sound source on the basis of ITD and motion parallax but
the localization was limited in a two-dimensional (2D) plane and was
not impressive under static conditions. Pang et al.~\cite{PangLiuZhangEtAl2017}
presented an approach for binaural azimuth estimation based on reverberation
weighting and generalized parametric mapping. Lu et al.~\cite{lu2007active}
presented a binaural distance localization approach using the motion-induced
rate of intensity change which requires the use of parallax motion
and errors up to 3.4 m were observed. Kneip and Baumann~\cite{LaurentKneip2008}
established formulae for binaural identification of the azimuth and
elevation angles as well as the distance information of a sound source
combining the rotational and translational motion of the interaural
axis. However, large localization errors were observed and no solution
was given to handle sensor noise nor model uncertainty. Rodemann~\cite{Rodemann2010}
proposed a binaural azimuth and distance localization technique using
signal amplitude along with ITD and ILD cues in an indoor environment
with a sound source ranging from $0.5\;\text{m}$ to $6\;\text{m}$.
However, the azimuth estimation degrades with the distance and reduced
error with the required calibration was still large. Kumon and Uozumi~\cite{kumon2011binaural}
proposed a binaural system on a robot to localize a mobile sound source
but it requires the robot to move with a constant velocity to achieve
2D localization. Also, further study was proposed for a parameter
$\alpha_{0}$ introduced in the EKF. Zhong et al.~\cite{Sun2015,zhong2016active}
and Gala et al.~\cite{Gala2018} utilized the extended Kalman filtering
(EKF) technique to perform orientation localization using the ITD
data acquired by a set of binaural self-rotating microphones. Moreover,
large errors were observed in~\cite{zhong2016active} when the elevation
angle of a sound source was close to zero. 

To the best of our knowledge, the works presented in the literature
for SSL using two microphones based on ITD cues mainly provided formulae
that calculate the azimuth and elevation angles of a sound source
without incorporating sensor noise~\cite{LaurentKneip2008}. The
works that use probabilistic recursive filtering techniques (e.g.,
EKF) for orientation estimation~\cite{zhong2016active} did not conduct
any observability analysis on the system dynamics. In other words,
no discussion on the limitation of the techniques for orientation
estimation was found. In addition, no probabilistic recursive filtering
technique was used to acquire distance information of a sound source.
This paper aims to address these research gaps. 

The contributions of this paper include (1) an observability analysis
of the system dynamics for three-dimensional (3D) SSL using two microphones
and the ITD cue only; (2) a novel algorithm that provides the estimation
of the elevation angle of a sound source when the states are unobservable;
and (3) a new EKF-based technique that estimates the robot-sound distance.
Both simulations and experiments were conducted to validate the proposed
techniques. 

The rest of this paper is organized as follows. Section~\ref{sec:PRELIMINARIES}
describes the preliminaries. In Section~\ref{sec:Angular-Localization},
2D and 3D orientation localization models are presented along with
their observability analysis. In Section~\ref{sec:Complete-Solution},
a novel method is proposed to detect non-observability conditions
and a solution to the non-observability problem is presented. Section~\ref{sec:Distance-Localization}
presents a distance localization model with its observability analysis.
The EKF algorithm is presented in Section~\ref{sec:EKF}. In Sections~\ref{sec:Simulation-Results}
and~\ref{sec:Experimental-Results}, the simulation and experimental
results are presented respectively, followed by Section~\ref{sec:Conclusion-and-Future},
which concludes the paper.

\section{PRELIMINARIES\label{sec:PRELIMINARIES}}

\subsection{Calculation of ITD}

The only cue used for localization in this paper is the ITD, which
is the time difference of a sound signal traveling to the two microphones
and can be calculated using the cross-correlation technique~\cite{Knapp1976The,azaria1984time}.
\begin{figure}
\hfill{}\includegraphics[width=0.98\columnwidth]{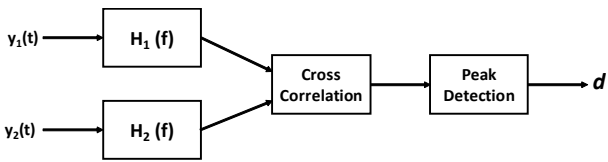}\hfill{}\caption{Interaural Time Delay (ITD) estimation between signals $y_{1}(t)$
and $y_{2}(t)$ using the cross-correlation technique. \label{fig:Delay estimation} }
\end{figure}

Consider a single stationary sound source placed in an environment.
Let $y_{1}(t)$ and $y_{2}(t)$ be the sound signals captured by two
spatially separated microphones in the presence of noise, which are
given by~\cite{Knapp1976The}
\begin{align}
y_{1}(t) & =s(t)+n_{1}(t),\\
y_{2}(t) & =\delta\cdot s\left(t+t_{d}\right)+n_{2}(t),
\end{align}
where $s\left(t\right)$ is the sound signal, $n_{1}(t)$ and $n_{2}(t)$
are real and jointly stationary random processes, \emph{$t_{d}$}
denotes the time difference of $s\left(t\right)$ arriving at the
two microphones, and $\delta$ is the signal attenuation factor due
to different traveling distances of the sound signal to the two microphones.
It is commonly assumed that $\delta$ changes slowly and \emph{$s(t)$}
is uncorrelated with noises $n_{1}(t)$ and $n_{2}(t)$~\cite{Knapp1976The}.
The cross-correlation function of $y_{1}(t)$ and $y_{2}(t)$ is given
by 
\[
R_{y_{1},y_{2}}(\tau)=E\left[y_{1}(t)\cdot y_{2}(t-\tau)\right],
\]
where \emph{$E\left[\cdot\right]$} represents the expectation operator.
Figure~\ref{fig:Delay estimation} shows the process of delay estimation
between $y_{1}(t)$ and $y_{2}(t)$, where $H_{1}(f)$ and $H_{2}(f)$
represent scaling functions or pre-filters~\cite{Knapp1976The}.
Various techniques can be used to eliminate or reduce the effect of
background noise and reverberations~\cite{Boll1979,BollPulsipher1980,naylor2010speech,spriet2007speech,Gala2010,Gala2011}.
An improved version of the cross-correlation method incorporating
$H_{1}(f)$ and $H_{2}(f)$ is called Generalized Cross-Correlation
(GCC)~\cite{Knapp1976The}, which further improves the estimation
of time delay. 

The time difference of $y_{1}(t)$ and $y_{2}(t)$, i.e., the ITD,
is given by $\hat{T}\triangleq\arg\max_{\tau}R_{y_{1},y_{2}}.$ The
distance difference of the sound signal traveling to the two microphones
is given by $d\triangleq\hat{T}\cdot c_{0},$ where $c_{0}$ is the
sound speed and is usually selected to be 345 m/s. 

\subsection{Far-Field Assumption\label{subsec:Far-Field-Assumption}}

The area around a sound source can be divided into five different
fields: free field, near field, far field, direct field and reverberant
field~\cite{ISO12001,Hansen2001}. The region close to a source where
the sound pressure and the acoustic particle velocity are not in phase
is regarded as the near field. The range of the near field is limited
to a distance from the source equal to approximately a wavelength
of sound or equal to three times the largest dimension of the sound
source, whichever is the larger. The far field of a source begins
where the near field ends and extends to infinity. Under the far-field
assumption, the acoustic wavefront reaching the microphones is planar
and not spherical, in the sense that the waves travel in parallel
i.e. the angle of incidence is the same for the two microphones~\cite{Calmes2009}.

\subsection{Observability Analysis}

Consider a nonlinear system described by a state-space model
\begin{align}
\dot{\mathbf{x}} & =f\left(\mathbf{x}\right),\label{eq:process}\\
\mathbf{y} & =h\left(\mathbf{x}\right),\label{eq:output}
\end{align}
where $\mathbf{x}\in\mathbb{R}^{n}$ and $\mathbf{y}\in\mathbb{R}^{m}$
are the state and output vectors, respectively, and $f\left(\cdot\right)$
and $h\left(\cdot\right)$ are the process and output functions, respectively.
The observability matrix of the system described by~(\ref{eq:process})
and~(\ref{eq:output}) is then given by~\cite{hedrick2005control}
\begin{align*}
\Omega & =\left[\begin{array}{ccc}
\left(\frac{\partial L_{f}^{0}h}{\partial\mathbf{x}}\right)^{T} & \left(\frac{\partial L_{f}^{1}h}{\partial\mathbf{x}}\right)^{T} & \cdots\end{array}\right]^{T},
\end{align*}
where the Lie derivatives are given by $L_{f}^{0}h=h\left(\mathbf{x}\right)$
and $L_{f}^{n}h=\frac{\partial L_{f}^{n-1}h}{\partial\mathbf{x}}f$.
The system is observable if the observability matrix $\Omega$ has
rank $n$. 

\section{Mathematical Models and Observability Analysis for Orientation Localization\label{sec:Angular-Localization}}

The complete localization of a sound source is usually achieved in
two stages, the orientation (i.e., azimuth and elevation angles) localization
and distance localization. In this section, the methodology of the
orientation localization is presented.
\begin{figure}[h]
\centering{}\includegraphics[width=0.7\columnwidth]{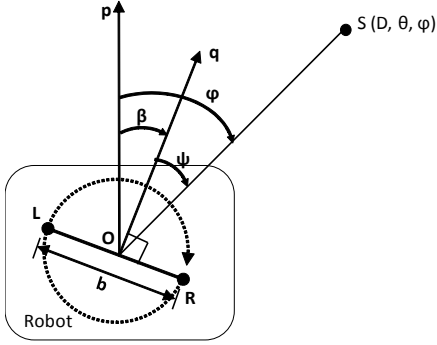}\caption{Top view of the robot illustrating different angle definitions due
to the rotation of the microphone array. \label{figure system formulation}}
\end{figure}
\begin{figure}[h]
\centering{}\includegraphics[width=0.7\columnwidth]{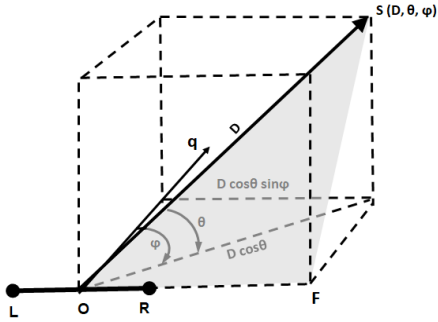}\caption{3D view of the system for orientation localization.\label{fig:Spherical Coordinate System}}
\end{figure}

\subsection{Definitions}

As shown in Figures~\ref{figure system formulation} and \ref{fig:Spherical Coordinate System},
the acoustic signal generated by the sound source \emph{$S$} is collected
by the left and right microphones, $L$ and $R$, respectively. Let
$O$ be the center of the robot as well as the two microphones. The
location of $S$ is represented by ($D,\theta,\varphi$), where $D$
is the distance between the source and the center of the robot, i.e.,
the length of segment $\overline{OS}$, $\theta\in\left[0,\frac{\pi}{2}\right]$
is the elevation angle defined as the angle between $\overline{OS}$
and the horizontal plane, and $\varphi\in\left(-\pi,\pi\right]$ is
the azimuth angle defined as the angle measured clockwise from the
robot heading vector, $\mathbf{p}$, to $\overline{OS}$. Letting
unit vector $\mathbf{q}$ be the orientation (heading) of the microphone
array, $\beta$ be the angle between $\mathbf{p}$ and $\mathbf{q}$,
and $\psi$ be the angle between $\mathbf{q}$ and $\overline{OS}$,
both following a right hand rotation rule, we have
\begin{equation}
\varphi=\psi+\beta.\label{eq:azimuth}
\end{equation}
For a clockwise rotation, we have $\beta\left(t\right)=\omega t$,
where $\omega$ is the rotational speed of the two microphones, and
$\psi\left(t\right)=\varphi-\omega t$. 
\begin{figure}[h]
\centering{}\includegraphics[width=0.7\columnwidth]{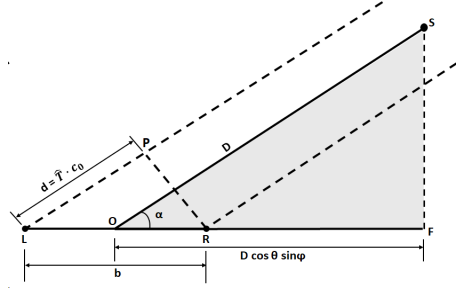}\caption{The shaded triangle in Figure~\ref{fig:Spherical Coordinate System}.\label{fig:Gray Triangle Plane}}
\end{figure}
In the shaded triangle, $\triangle SOF$, shown in Figures\emph{\large{}~}\ref{fig:Spherical Coordinate System}
and \ref{fig:Gray Triangle Plane}, define $\alpha\triangleq\angle SOF$
and we have $\alpha+\psi=\frac{\pi}{2}$ and $\cos\alpha=\cos\theta\sin\psi.$
Based on the far-field assumption in Section~\ref{subsec:Far-Field-Assumption},
we have
\begin{equation}
d=\hat{T}\cdot c_{0}=b\cos\alpha=b\cos\theta\sin\psi.\label{eq:d-1}
\end{equation}
where $b$ is the distance between the two microphones, i.e. the length
of the segment $\overline{LR}$. 

To avoid cone of confusion~\cite{Wallach1939} in SSL, the two-microphone
array is rotated with a nonzero angular velocity~\cite{zhong2016active}.
Without loss of generality, in this paper we assume a clockwise rotation
of the microphone array on the horizontal plane while the robot itself
does not rotate nor translate throughout the entire estimation process,
which implies that $\varphi$ is constant. 

\subsection{2D Localization}

If the sound source and the robot are on the same horizontal plane,
i.e., $\theta=0$, we have $d=b\sin\psi$. Assume that the microphone
array rotates clockwise with a constant angular velocity, $\omega$.
Considering the state-space model for 2D localization with the state
$x_{2D}\triangleq\psi$, and the output as $y_{2D}\triangleq d$,
we have
\begin{align}
\dot{x}_{2D} & =\dot{\psi}=-\omega,\label{eq:2D state xdot}\\
y_{2D} & =b\sin\psi.\label{eq:2D state y}
\end{align}
\begin{theorem}
\label{thm:2D}The system described by Equations~(\ref{eq:2D state xdot})
and (\ref{eq:2D state y}) is observable if 1)~$b\neq0$ and 2)~$\omega\neq0$
or $\psi\neq2k\pi\pm\frac{\pi}{2}$, where $k\in\mathbb{Z}$.
\end{theorem}
\begin{proof}
The observability matrix~\cite{hermann1977nonlinear,hedrick2005control}
for the system described by Equations~(\ref{eq:2D state xdot}) and
(\ref{eq:2D state y}) is given by
\begin{equation}
O_{2D}=\left[\begin{array}{cccc}
b\cos\psi & b\omega\sin\psi & -b\omega^{2}\cos\psi & \cdots\end{array}\right]^{T}.
\end{equation}
The system is observable if $O_{2D}$ has rank one, which implies
$b\neq0$. If $\omega=0$, observability requires that $\cos\psi\neq0$,
which implies $\psi\neq2k\pi\pm\frac{\pi}{2}$ . If $\omega\neq0$,
$O_{2D}$ is full rank for all $\psi$. 
\end{proof}

\begin{remark}
\label{rem:Since-the-microphones}Since the two microphones are separated
by a non-zero distance, (i.e., $b\neq0$) and the microphone array
rotates with a non-zero constant angular velocity (i.e., $\omega\neq0$),
the system is observable in the domain of definition. 
\end{remark}

\subsection{3D Localization}

Considering the state-space model for 3D localization with the state
$x_{3D}\triangleq\left[\theta,\psi\right]^{T}$, and the output as
$y_{3D}\triangleq d$, we have
\begin{align}
\dot{x}_{3D} & =\left[\begin{array}{c}
\dot{\theta}\\
\dot{\psi}
\end{array}\right]=\left[\begin{array}{c}
0\\
-\omega
\end{array}\right],\label{eq:3D state x}\\
y_{3D} & =b\cos\theta\sin\psi.\label{eq:3d_output}
\end{align}
\begin{theorem}
\label{thm:3D}The system described by Equations~(\ref{eq:3D state x})
and~(\ref{eq:3d_output}) is observable if 1)~$b\neq0$, 2)~$\omega\neq0$,
3)~$\theta\neq0^{o}$, and 4)~$\theta\neq90^{o}$.
\end{theorem}
\begin{proof}
The observability matrix for~(\ref{eq:3D state x}) and~(\ref{eq:3d_output})
is given by 
\begin{equation}
O_{3D}=\left[\begin{array}{cc}
-b\sin\theta\sin\psi & b\cos\theta\cos\psi\\
b\omega\sin\theta\cos\psi & b\omega\cos\theta\sin\psi\\
b\omega^{2}\sin\theta\sin\psi & -b\omega^{2}\cos\theta\cos\psi\\
-b\omega^{3}\sin\theta\cos\psi & -b\omega^{3}\cos\theta\sin\psi\\
\cdots & \cdots
\end{array}\right].
\end{equation}
It should be noted that higher-order Lie derivatives do not add rank
to $O_{3D}$. Consider the squared matrix consisting of the first
two rows of $O_{3D}$
\[
\Omega_{3D}=\left[\begin{array}{cc}
-b\sin\theta\sin\psi & b\cos\theta\cos\psi\\
b\omega\sin\theta\cos\psi & b\omega\cos\theta\sin\psi
\end{array}\right],
\]
and the determinant of the $\Omega_{3D}$ is 
\[
\det\left\{ \Omega_{3D}\right\} =-b^{2}\omega\sin\theta\cos\theta.
\]
The system is observable if
\begin{equation}
b\neq0,\text{ }\omega\neq0,\text{ }\theta\neq0^{o},\text{ }and\text{ }\theta\neq90^{o}.
\end{equation}
Further investigation can be done by selecting two even (or odd) rows
from $O_{3D}$ to form a squared matrix, whose determinant is always
zero. .
\end{proof}

\begin{remark}
As it is always true that $b\neq0$ and $\omega\neq0$ due to Remark~\ref{rem:Since-the-microphones},
the system is observable only when $\theta\neq0^{o}$ and $\theta\neq90^{o}$.
Experimental results presented by Zhong et al.~\cite{zhong2016active}
using a similar model illustrates large estimation error when $\theta$
is close to zero. 
\end{remark}
To further investigate the system observability, consider the following
two special cases: (1) $\theta$ is known and (2) $\psi$ is known. 

Assume that $\theta$ is known and consider the following system 
\begin{align}
\dot{x}_{\psi} & =\dot{\psi}=-\omega,\label{eq:3D state psi xdot}\\
y_{\psi} & =b\cos\theta\sin\psi.\label{eq:3D state psi y}
\end{align}
\begin{corollary}
\label{cor:AzimuthObservability}The azimuth angle in the system described
by Equations~(\ref{eq:3D state psi xdot}) and~(\ref{eq:3D state psi y})
is observable if 1)~$b\neq0$, 2)~$\omega\neq0$, and 3)~$\theta\neq90^{o}$.
\end{corollary}
\begin{proof}
The observability matrix associated with~(\ref{eq:3D state psi xdot})
and~(\ref{eq:3D state psi y}) is given by 
\begin{equation}
O_{\psi}=\left[\begin{array}{ccc}
b\cos\theta\cos\psi & b\omega\cos\theta\sin\psi & \cdots\end{array}\right]^{T}.
\end{equation}
So, the system is observable if,
\begin{equation}
b\neq0,\;\theta\neq90^{o},\;\text{and}\;\omega\neq0\;\text{or}\;\psi\neq2k\pi\pm\frac{\pi}{2}.
\end{equation}
This shows that $\psi$ is unobservable when $\theta=90^{o}$.
\end{proof}

Assume that $\psi$ is known and consider the following system
\begin{align}
\dot{x}_{\theta} & =\dot{\theta}=0,\label{eq:3D state theta xdot}\\
y_{\theta} & =b\cos\theta\sin\psi.\label{eq:3D state theta y}
\end{align}
\begin{corollary}
\label{cor:ElevationObservability}The elevation angle in the system
described by Equations~(\ref{eq:3D state theta xdot}) and~(\ref{eq:3D state theta y})
is observable if the following conditions are satisfied: 1)~$b\neq0$,
2)~$\omega\neq0$, and 3)~$\theta\neq0^{o}$.
\end{corollary}
\begin{proof}
The observability matrix asscoiated with~(\ref{eq:3D state theta xdot})
and~(\ref{eq:3D state theta y}) is given by 
\begin{equation}
O_{\theta}=\left[\begin{array}{ccc}
-b\sin\theta\sin\psi & 0 & \cdots\end{array}\right]^{T}.
\end{equation}
So the system is observable if
\begin{equation}
b\neq0,\;\theta\neq0^{o},\;\text{and}\;\psi\neq k\pi.
\end{equation}
As $\psi$ is time-varying, so it won't stay at $k\pi$. It can be
seen that $\theta$ is unobservable when $\theta=0^{o}$.
\end{proof}

\section{Complete Orientation Localization\label{sec:Complete-Solution}}

To handle the unobservable situations, i.e., $\theta=0^{o}$ and $\theta=90^{o}$,
we present a novel algorithm in this section that utilizes both the
2D and 3D localization models to enable the orientation localization
of a sound source residing anywhere in the domain of definition, i.e.,
$\theta\in\left[0,\pi/2\right]$ and $\varphi\in\left(-\pi,\pi\right]$.
\begin{figure}
\centering\includegraphics[clip,width=0.98\columnwidth]{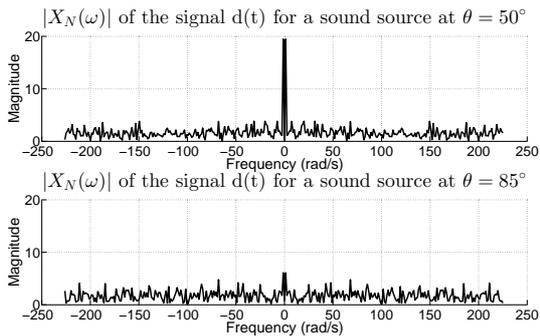}

\caption{The signals after taking the Discrete Fourier transform (DFT) of the
noised signal $d(t)$ with the source located at $\theta=50^{o}$
and $\theta=85^{o}$, respectively. The two big peaks in the upper
figure occur at $\pm2\pi/5$ rad/sec (i.e., the angular velocity of
the rotation of the microphone array) when $\theta=50^{o}$, whereas
small peaks are present in the bottom figure when $\theta=85^{o}$.\label{fig:FFT_Magnitude}}

\end{figure}
\begin{figure}[h]
\centering\includegraphics[width=0.98\columnwidth]{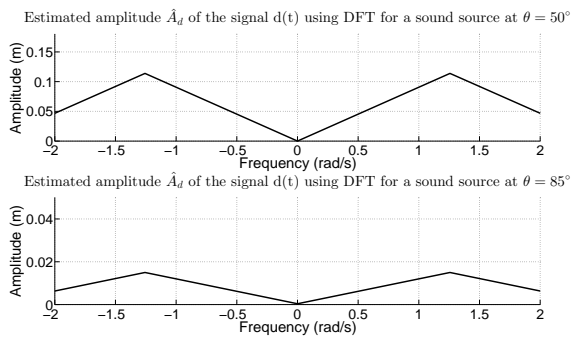}

\caption{Estimated amplitudes $\hat{A}_{d}$ of the signal $d(t)$ using the
DFT for the source located at $\theta=50^{o}$ and $\theta=85^{o}$,
respectively. The maximum of $\hat{A}_{d}$ occurs when the frequency
is at $\pm2\pi/5$ rad/sec. When $\theta=85^{o}$, the maximum of
$\hat{A}_{d}$ is less than 0.017 m. \label{fig:FFT_Amplitude} }
\end{figure}

\subsection{Identification of $\theta=90^{o}$ \label{subsec:CheckIdentification-of-}}

ITD could be zero due to either $90^{o}$ elevation or absence of
sound, the latter of which can be detected by evaluating the power
reception of microphones. In this paper, we focus on the former case. 

Assume that the sensor noise is Gaussian, which dominates the ITD
signal when $\theta$ gets close to $90^{o}$. To check the presence
of the signal $d(t)$ buried in the noise, we can first apply the
Discrete Fourier Transform (DFT) onto the stored $d\left(t\right)$.
The $N$-point DFT of the signal $d(t)$ results in a sequence of
complex numbers in the form of $X_{real}+jX_{imag}$, where $X_{real}$
and $X_{imag}$ represent the real and imaginary coordinates of the
complex number. The magnitude of the complex number is then obtained
by $|X(\omega)|=\sqrt{X_{real}^{2}+X_{imag}^{2}}$. Figure~\ref{fig:FFT_Magnitude}
shows the resulting magnitude ($|X(\omega)|$) signals of $d(t)$
after taking DFT when the sound source is placed at $\theta=50^{o}$
and $85^{o}$, respectively, in simulation. Two big peaks in the top
subfigure (i.e., when $\theta=50^{o}$) are observed when the frequency
is at $\pm2\pi/5$ rad/sec (i.e., the angular velocity of the rotation
of the microphone array). However, the peaks observed in the bottom
subfigure (i.e., when $\theta=85^{o}$) are comparatively very small. 

To eliminate the noise in Figure~\ref{fig:FFT_Magnitude}, define
the estimated amplitude of the ITD signal as $\hat{A}_{d}(\omega)=\frac{2}{N}\cdot|X(\omega)|$.
Figure~\ref{fig:FFT_Amplitude} shows the estimated amplitude ($\hat{A}_{d}$)
of the signal $d(t)$ resulting from Figure~\ref{fig:FFT_Magnitude}.
The bottom subfigure (i.e., when $\theta=85^{o}$) shows that the
maximum value of $\hat{A}_{d}$ is very small compared to the top
subfigure (i.e., when $\theta=50^{o}$). The ITD is considered as
zero if the maximum value of the estimated amplitude $\hat{A}_{d}$
(when the frequency equals the angular velocity of the rotation of
the microphone array) is less than a predefined threshold, $d_{threshold}$.
The selection of $d_{threshold}$ determines the accuracy of the estimation
when the sound source is around $90^{o}$ elevation. The value of
$d_{threshold}$, for example, can be selected as $0.017$ m, which
corresponds to $\theta=85^{o}$ as in Figure~\ref{fig:FFT_Amplitude},
thereby giving an accuracy of $5^{o}$. 

\subsection{Identification of $\theta=0^{o}$ \label{subsec:Theta Zero Deg Identification}}

Theorem~\ref{thm:2D} guarantees accurate azimuth angle estimation
using the 2D model when the sound source is located with zero elevation.
We observed that when the elevation of the sound source is not close
to zero, the estimation of the azimuth angle provided by the 2D model
is far off the real value. 
\begin{figure}[h]
\centering\includegraphics[width=0.98\columnwidth]{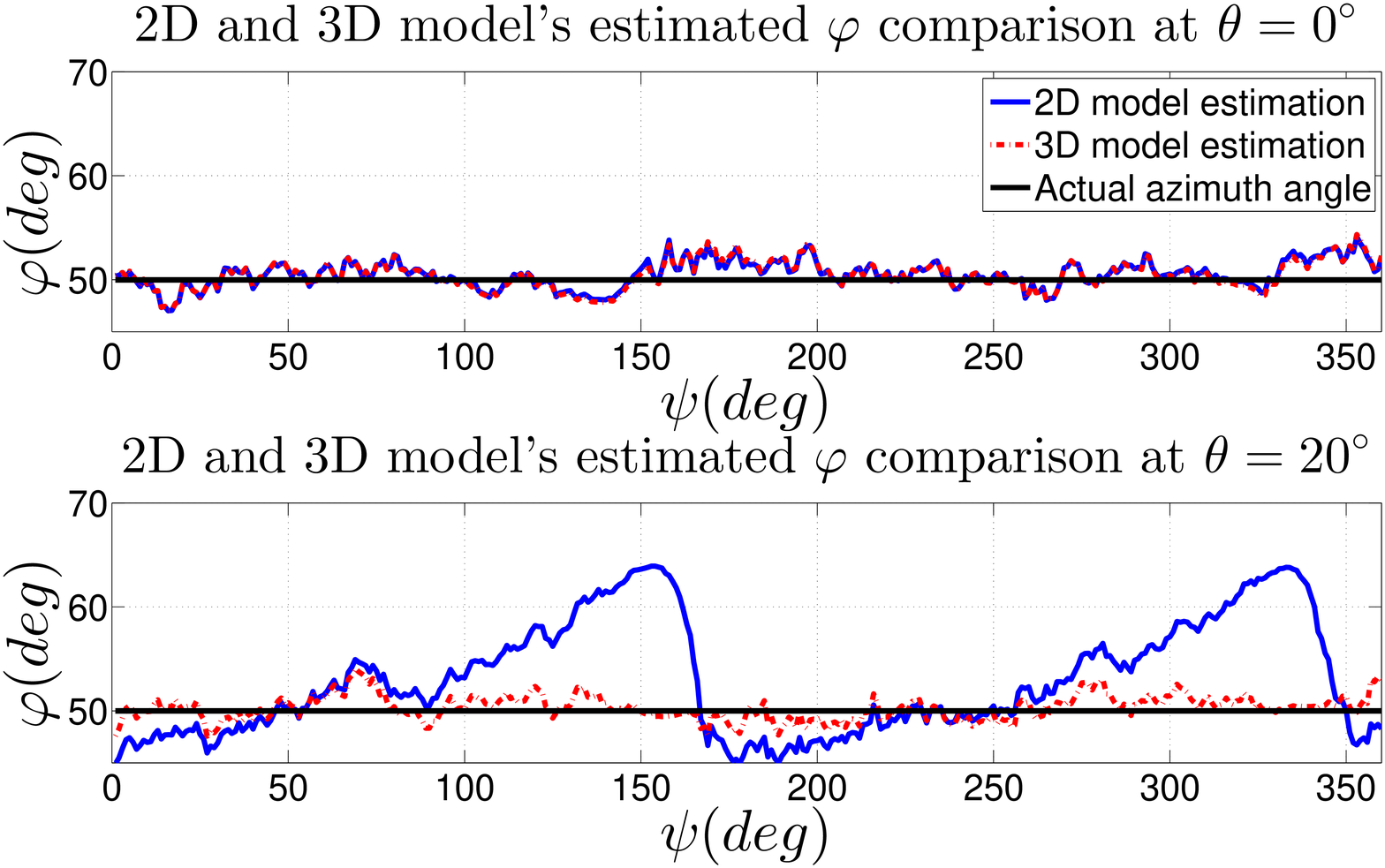}\caption{Comparison of azimuth angle estimations using the 2D and 3D localization
models when a sound source is located at $\theta=0^{o}$ and $\theta=20^{o}$,
respectively. \label{fig:AzimuthComparisonSubplot}}
\end{figure}

On the other hand, Theorem~\ref{thm:3D} guarantees that the azimuth
angle estimation using the 3D model is accurate for all elevation
angles except for $\theta=90^{o}$, which is detected by the approach
in Section \ref{subsec:CheckIdentification-of-}. Therefore, the estimations
resulting from both the 2D model 3D models will be identical if the
sound source is located at $\theta=0^{o}$, as shown in Figure\emph{\large{}~}\ref{fig:AzimuthComparisonSubplot}.
The root-mean-square error (RMSE) is used as a measure of the difference
between the two azimuth estimations as it includes both mean absolute
error (MAE) as well as additional information related to the variance~\cite{brassington2017mean}.
This error is dependent on the value of elevation angle and it increases
as the elevation angle increases, as shown in Figure~\ref{fig:RMSE Surf}.

\begin{figure}[h]
\centering\includegraphics[width=0.98\columnwidth]{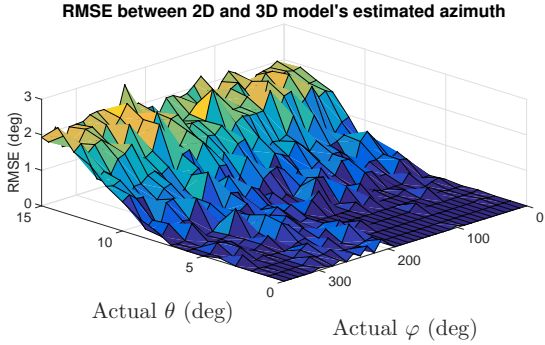}\caption{RMSE between 2D and 3D localization model's estimated azimuth angles.\label{fig:RMSE Surf}}
\end{figure}
\begin{figure}[h]
\centering\includegraphics[width=0.98\columnwidth]{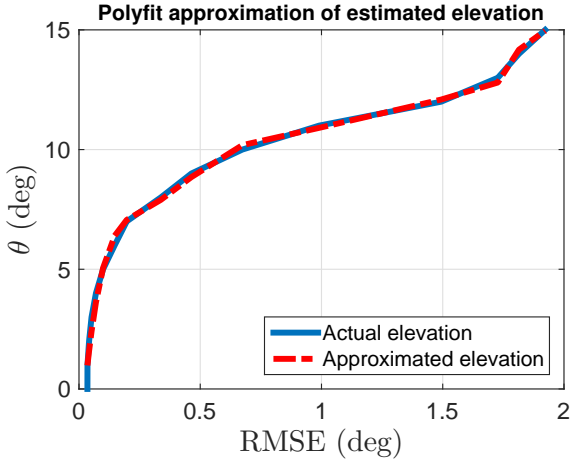}\caption{Approximation of the elevation angle from the RMSE data using the
least square fitted polynomial. \label{fig:Polyfit}}
\end{figure}
In order to get an accurate estimate of the elevation angle close
to zero, a polynomial curve fitting approach is used to map (in a
least-square sense) the RMSE values to the elevation angles. Different
RMSE values are collected beforehand in the environment where the
localization would be done. The RMSE values associated with the same
elevation angle but different azimuth angles express small variations,
as seen in Figure\emph{\large{}~}\ref{fig:RMSE Surf}. Therefore,
for a particular elevation angle, the mean of all RMSE values with
different azimuth angles will be selected as the RMSE value corresponding
to the elevation angle. An example curve is shown in Figure~\ref{fig:Polyfit}. 

\subsection{Complete Orientation Localization Algorithm}

\begin{figure}[h]
\centering{}\includegraphics[width=0.98\columnwidth]{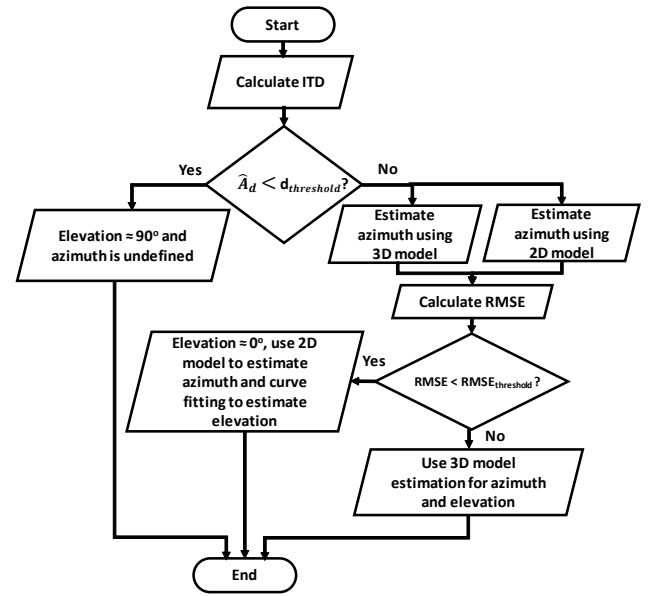}\caption{Flowchart addressing the non-observability problems reflected in the
3D localization model. \label{fig:Flowchart combined model}}
\end{figure}
\begin{algorithm}
1: Calculate the ITD, $\hat{T}$, from the recorded signals of two
microphones.

2: \textbf{IF} $\hat{A}_{d}<d_{threshold}$ \textbf{THEN}

3: $\text{ \text{ \text{ \text{ }}}}$The elevation angle of the sound
source is $\theta=90^{o}$ and the azimuth angle, $\varphi$, is undefined.

4: \textbf{ELSE}

5: $\text{ \text{ \text{ \text{ }}}}$Estimate the azimuth $\varphi_{2D}$
and $\varphi_{3D}$ using 2D and 3D localization models, respectively

6: $\text{ \text{ \text{ \text{ }}}}$Calculate the $RMSE$ between
$\varphi_{2D}$ and $\varphi_{3D}$

7: $\text{ \text{ \text{ \text{ }}}}$\textbf{IF }$RMSE<RMSE_{threshold}$
\textbf{THEN}

8: $\text{ \text{ \text{ \text{ }}}}$$\text{ \text{ \text{ \text{ }}}}$Use
polynomial curve fitting to determine $\theta$ using the calculated
$RMSE$ value and estimate $\varphi$ using either 2D or 3D localization
model

9: $\text{ \text{ \text{ \text{ }}}}$\textbf{ELSE }estimate both
$\theta$ and $\varphi$ using the 3D localization model

10:$\text{ \text{ \text{ \text{ }}}}$\textbf{END IF}

11: \textbf{END IF}

\caption{Complete 3D orientation localization\label{alg:Algorithm combined model}}

\end{algorithm}
Figure\emph{\large{}~}\ref{fig:Flowchart combined model} illustrates
the flowchart of the proposed algorithm for the complete orientation
localization. The pseudo code of the proposed complete orientation
localization is given in Algorithm~\ref{alg:Algorithm combined model}.
The $RMSE_{threshold}$ is the value used to check when the elevation
angle is close to $0^{o}$. This threshold value decides the point
until which the curve fitting is required, ansd after which the 3D
model can be trusted for elevation estimation. 

\section{Distance Localization\label{sec:Distance-Localization}}

\begin{figure}[h]
\begin{centering}
\includegraphics[width=0.7\columnwidth]{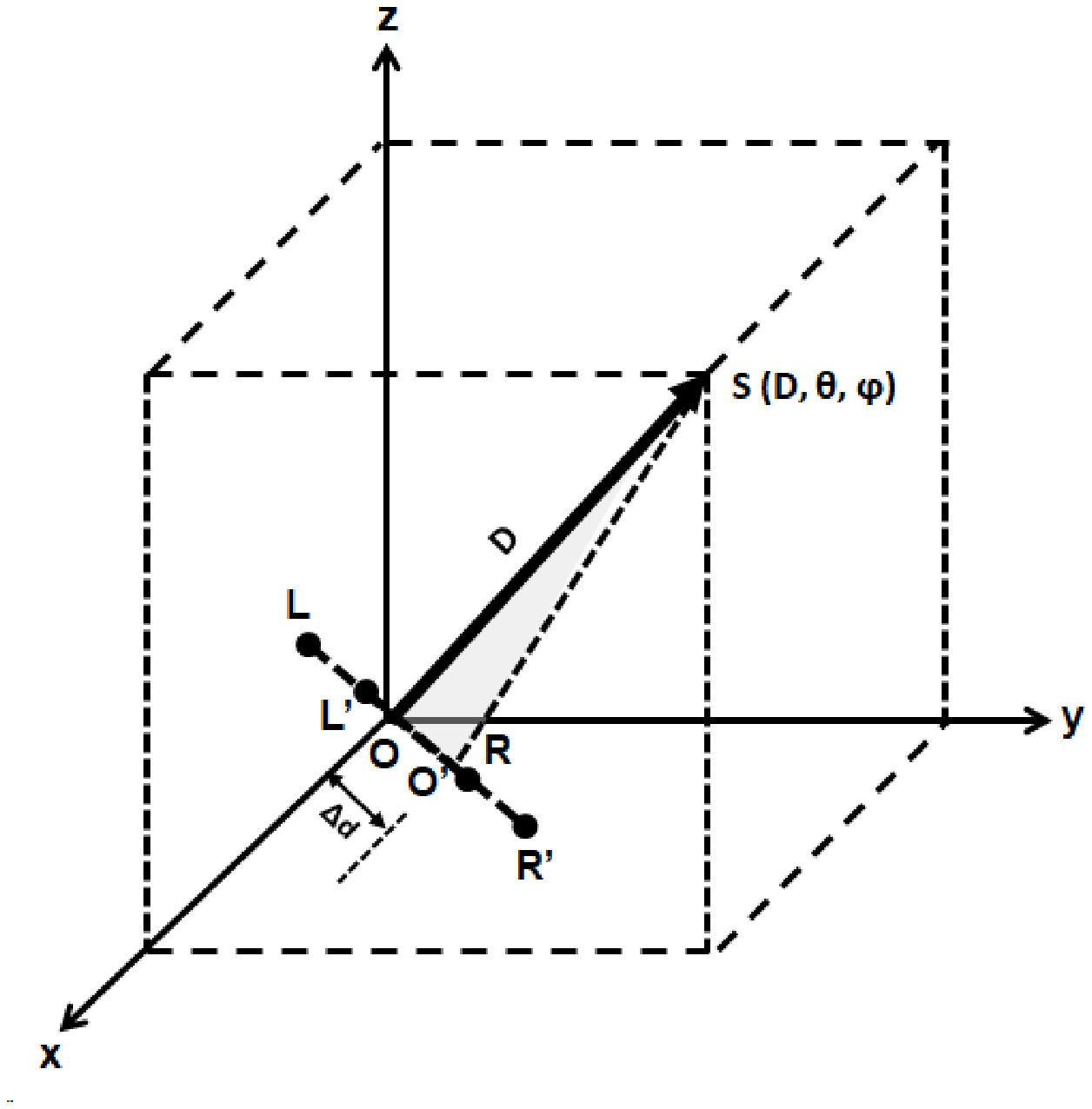}
\par\end{centering}
\caption{3D view of the system for distance localization.\label{fig:Distance Estimation}}
\end{figure}
\begin{figure}[h]
\begin{centering}
\includegraphics[width=0.5\columnwidth]{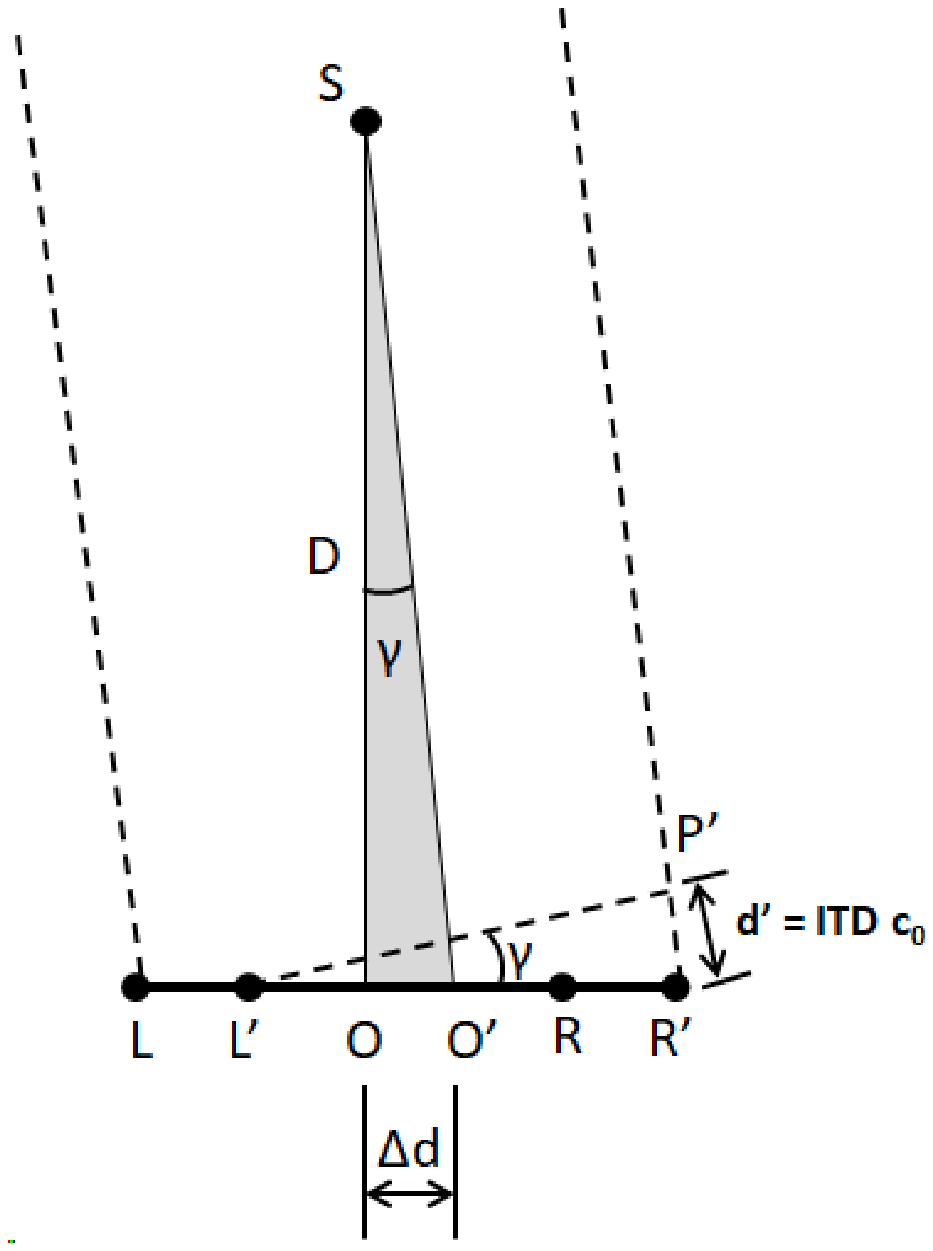}
\par\end{centering}
\caption{Gray triangle in Figure~\ref{fig:Distance Estimation}.\label{fig:Gray Triangle Distance}}
\end{figure}

The novel distance localization approach presented in this section
depends on an accurate orientation localization. Assume that the angular
location of the sound source has been obtained by using Algorithm~\ref{alg:Algorithm combined model}
and the microphone array has been regulated facing toward the sound
source, as shown in Figure~\ref{fig:Distance Estimation}. The proposed
distance localization approach requires the microphone array, $\overline{LR}$,
to translate with a distance $\Delta d$ along the line perpendicular
to the center-source vector (on the horizontal plane). This translation
shifts the center of the microphone array, $O$, to a new point, \emph{$O'$,
}and $\gamma$ is defined as the angle between vectors \emph{$\overline{O'S}$}
and $\overline{OS}$, as shown in Figure~\ref{fig:Gray Triangle Distance}.
Note that the center of the robot, \emph{O, }is unchanged.\emph{ }The
objective is to estimate distance \emph{D }between the center of the
robot \emph{O} and the source \emph{S.} 

\subsection{Mathematical Model for Distance Localization}

Consider the gray triangle shown in Figure~\ref{fig:Gray Triangle Distance}.
Based on the far-field assumption in Section~\ref{subsec:Far-Field-Assumption},
the length $\overline{R'P'}$ is given by
\begin{equation}
d'=b\sin\gamma.\label{eq:ITD distance}
\end{equation}
 In triangle $\triangle SOO'$, we have
\begin{equation}
\sin\gamma=\frac{\triangle d}{\sqrt{(\triangle d)^{2}+D^{2}}}.\label{eq:sin beta}
\end{equation}
Defining the state as $x_{dist}=D$ and output as $y_{dist}$, the
state-space model is given by
\begin{align}
\dot{x}_{dist} & =0,\label{eq:Distance state x}\\
y_{dist} & =\frac{b\text{ }\triangle d}{\sqrt{(\triangle d)^{2}+D^{2}}}.\label{eq:Distance state y}
\end{align}
\begin{theorem}
The system described by Equations~(\ref{eq:Distance state x})~and~(\ref{eq:Distance state y})
is observable if the following conditions are satisfied: 1)~$b\neq0$,
2)~$\triangle d\neq0$, and 3)~$D\neq0$.
\end{theorem}
\begin{proof}
The observability matrix associated with (\ref{eq:Distance state x})
and~(\ref{eq:Distance state y}) is given by 
\begin{equation}
O_{dist}=\left[\begin{array}{cc}
\frac{-2\text{ }b^{2}\text{ }(\triangle d)^{2}\text{ }D}{(\triangle d)^{2}+D^{2}} & \cdots\end{array}\right].
\end{equation}
So the system is observable if
\begin{equation}
b\neq0,\text{ }\triangle d\neq0,\text{ }and\text{ }D\neq0.
\end{equation}
\end{proof}

\begin{remark}
As the microphones are separated by a non-zero distance, i.e., $b\neq0,$
and the microphone array is being translated by a non-zero distance,
i.e., $\triangle d\neq0$, the system is always observable unless
the sound source and the robot are at same location making $D=0$,
which is not in the scope of discussion of this paper.
\end{remark}

\section{Extended Kalman Filter\label{sec:EKF}}

\begin{algorithm}[h]
1: Initialize: $\hat{x}$

2: At each value of sample rate $T_{out}$,

3: \textbf{FOR} i = 1 to N \textbf{DO} 

$\text{ \text{ \text{ \text{ }}}}$ $\text{ \text{ \text{ \text{ }}}}$
\textbf{Prediction}

4: $\text{ \text{ \text{ \text{ }}}}$ $\hat{x}=\hat{x}+(\frac{T_{out}}{N})\text{ }f(\hat{x},u)$

5: $\text{ \text{ \text{ \text{ }}}}$ $A_{J}=\frac{\partial f}{\partial x}(\hat{x},u)$

6: $\text{ \text{ \text{ \text{ }}}}$ $P=P+(\frac{T_{out}}{N})(A_{J}P+PA_{J}^{T}+Q)$

$\text{ \text{ \text{ \text{ }}}}$ $\text{ \text{ \text{ \text{ }}}}$\textbf{Update}

7: $\text{ \text{ \text{ \text{ }}}}$ $C_{J}=\frac{\partial h}{\partial x}(\hat{x},u)$

8: $\text{ \text{ \text{ \text{ }}}}$ $K=PC_{J}^{T}(R+C_{J}PC_{J}^{T})^{-1}$

9: $\text{ \text{ \text{ }}}$ $\text{ }$$P=(I-KC_{J})P$

10:$\text{ \text{ \text{ \text{ }}}}$ $\hat{x}=\hat{x}+K(y[n]-h(\hat{x},u[n])$

11:\textbf{ END FOR}

\caption{Pseudo code for EKF~\cite{BeardMcLain2012} \label{alg:EKF pseudo code}}
\end{algorithm}
\begin{table*}[h]
\caption{EKF parameters.\label{tab:EKFparameters}}

\hfill{}%
\begin{tabular}{|c|c|c|}
\hline 
\multirow{2}{*}{Parameter} & Angular & Distance\tabularnewline
 & localization & localization\tabularnewline
\hline 
\hline 
Process noise variance ($\sigma_{vi},$$i=1,2$) & $0.01$ & $0.1$\tabularnewline
\hline 
Sensor noise variance ($\sigma_{w}$) & $0.01$ & $0.001$\tabularnewline
\hline 
Initial azimuth angle estimate ($\varphi_{initial}$) & $5^{o}$ & \textbf{\textendash{}}\tabularnewline
\hline 
Initial elevation angle estimate ($\theta_{initial}$) & $5^{o}$ & \textbf{\textendash{}}\tabularnewline
\hline 
Initial distance estimate ($D_{initial}$) & \textbf{\textendash{}} & $1$ m\tabularnewline
\hline 
\end{tabular}\hfill{}
\end{table*}
\begin{figure}[h]
\centering{}\includegraphics[width=0.98\columnwidth]{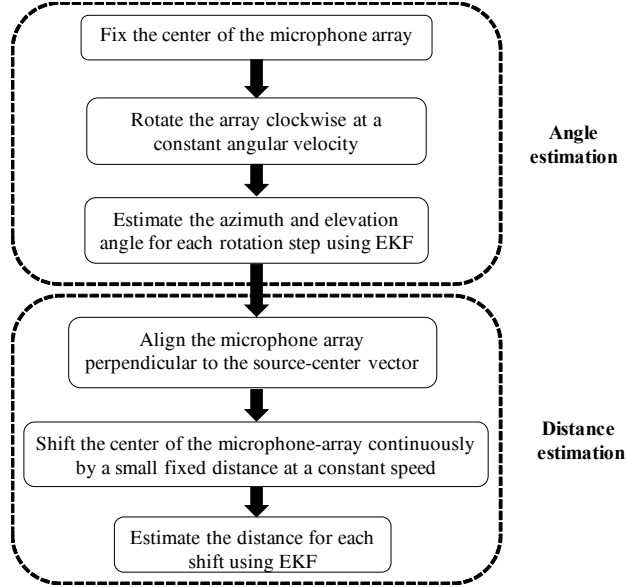}\caption{Block diagram showing the process for the proposed complete angular
and distance localization of a sound source using successive rotational
and translational motions of a set of two microphones. \label{fig:Localization algorithm}}
\end{figure}

The estimation for the angles and distance of the sound source is
conducted by extended Kalman filters. Detailed mathematical derivation
of the EKF can be found in~\cite{BeardMcLain2012}. Algorithm\emph{\large{}~}\ref{alg:EKF pseudo code}
summaries the EKF procedure used in this paper for SSL. The sensor
covariance matrix ($R$) is defined as $\sigma_{w}^{2}$, and the
process covariance matrix ($Q$) is defined as $\sigma_{v}^{2}$ for
the distance localization, $\sigma_{v1}^{2}$ for the 2D orientation
localization and $\text{diag}\{\sigma_{v1}^{2},\sigma_{v2}^{2}\}$
for the 3D orientation localization, respectively, where $\sigma_{vi}$
is the process noise variance corresponding to the $i^{th}$ state
and $\sigma_{w}$ is the sensor noise variance. Key parameters are
listed in Table~\ref{tab:EKFparameters}. The complete EKF-based
SSL procedure is illustrated in Figure~\ref{fig:Localization algorithm}. 

\section{Simulation Results\label{sec:Simulation-Results}}

In this section, we present the simulation results of the proposed
localization technique for both angle and distance localization of
a sound source. 

\subsection{Simulation Environment}

The Audio Array Toolbox~\cite{Donohue2009} is used to simulate a
rectangular space using the image method described in~\cite{allen1979image}.
The robot was placed in the center (origin) of the room. The two microphones
were separated by a distance of $0.18\;\text{m}$ from each other
which is equal to the approximate distance between human ears. The
sound source and the microphones are assumed omnidirectional and the
attenuation of the sound is calculated as per the specifications in
Table\emph{\large{}~}\ref{tab:Room specifications-1}.
\begin{table}[h]
\caption{Simulated room specifications \label{tab:Room specifications-1}}

\hfill{}%
\begin{tabular}{|c|c|}
\hline 
Parameter & Value\tabularnewline
\hline 
\hline 
Dimension & 20m x 20m x 20m\tabularnewline
\hline 
Reflection coefficient of each wall & 0.5\tabularnewline
\hline 
Reflection coefficient of the floor & 0.5\tabularnewline
\hline 
Reflection coefficient of the ceiling & 0.5\tabularnewline
\hline 
Velocity of the sound & 345 m/s\tabularnewline
\hline 
Temperature & 22$^{o}C$\tabularnewline
\hline 
Static pressure & 29.92 mmHg\tabularnewline
\hline 
Relative humidity & 38 \%\tabularnewline
\hline 
\end{tabular}\hfill{}
\end{table}

\subsection{Validation of Observablity}

As discussed earlier, Theorem~\ref{thm:2D} shows that the 2D model
is always observable, however, it does not provides any elevation
information of the sound source. On the other hand, Theorem~\ref{thm:3D}
shows that the 3D model is unobservable when the elevation angle of
the sound source is $0^{o}$ or $90^{o}.$ In order to validate the
observability analysis, localization was performed in the simulated
environment. 

For a sound source located on a 2D plane, Figure\emph{\large{}~}\ref{fig:Error2D}
shows the average of absolute estimation errors versus different azimuth
angles with the sound source at distance of $5\;\text{m}$ and $10\;\text{m}$
to the robot, respectively. It can be seen that all errors are smaller
than $1.8^{o}$ and the mean of the average of absolute errors is
approximately $1^{o}$ for the two cases. 
\begin{figure}
\begin{centering}
\includegraphics[width=0.98\columnwidth]{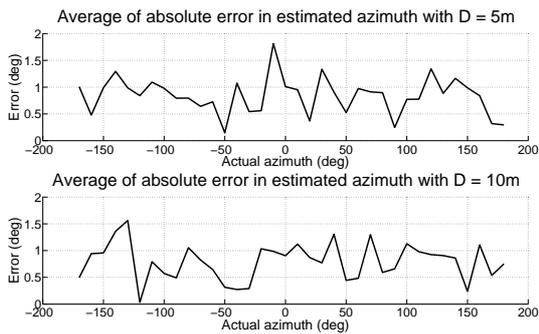}
\par\end{centering}
\caption{Average of absolute errors in azimuth angle estimation using 2D model
with a sound source placed at different azimuth locations at a constant
distance of 5 m and 10 m from the center of the robot.\label{fig:Error2D}}
\end{figure}
\begin{figure}
\hfill{}\includegraphics[width=0.98\columnwidth]{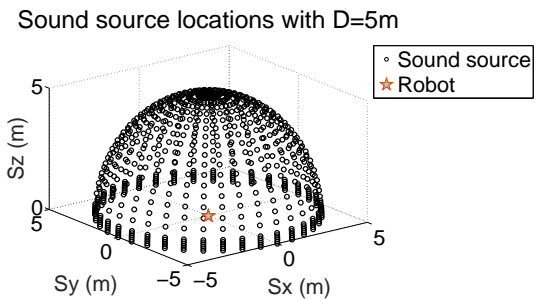}\hfill{}\caption{Sound source locations with a fixed distance of 5 m to the center
of the robot in the simulated room.\label{fig:SoundSourceLocations}}
\end{figure}
\begin{figure}
\centering\includegraphics[width=0.98\columnwidth]{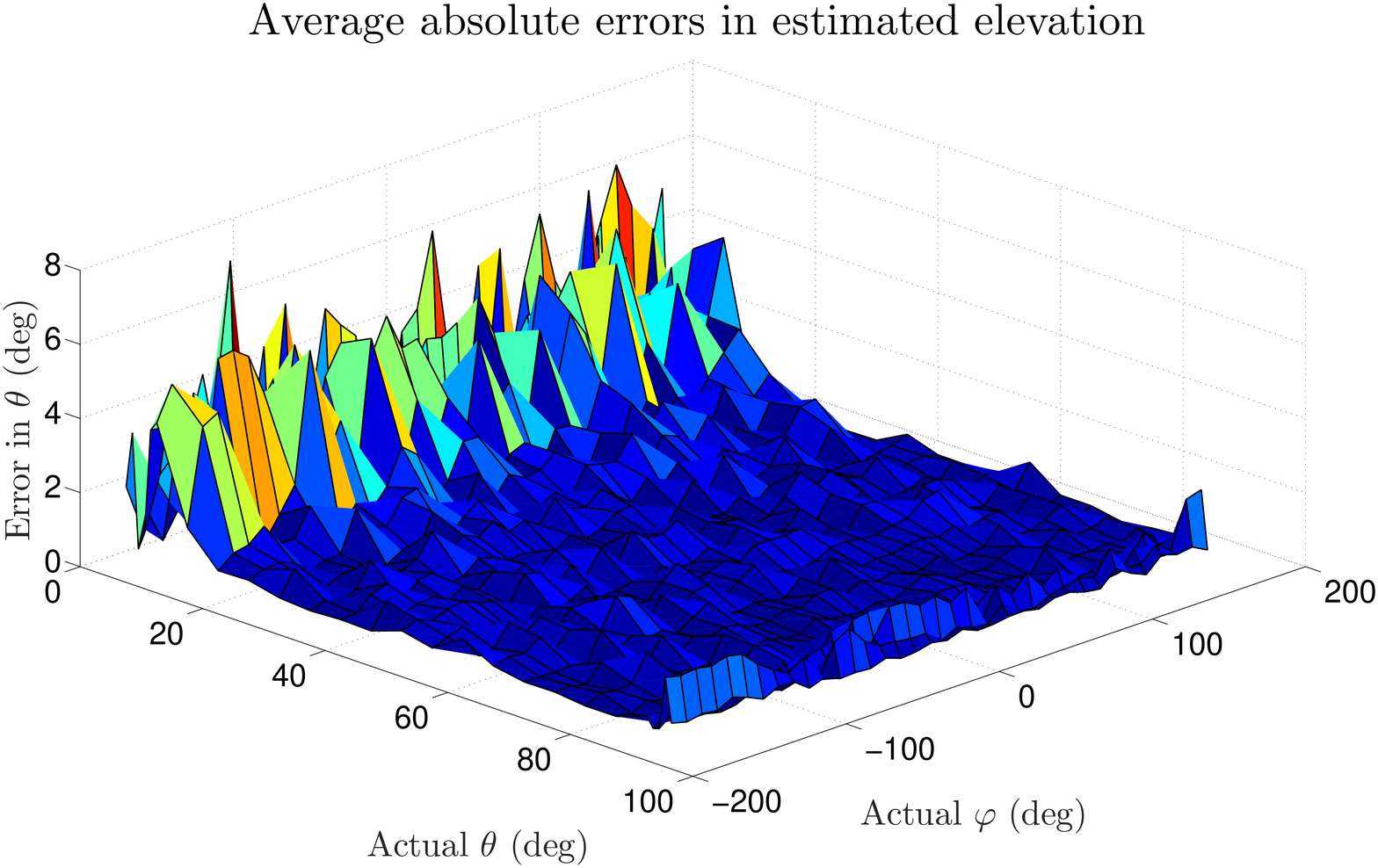}\caption{Average of absolute errors in elevation estimation using the 3D localization
model. Relatively large errors illustrate the non-observability condition
in elevation angle estimation with sound source placed around $0^{o}$
elevation, as described by Theorem~\ref{thm:3D} .\label{fig:ElevationError} }
\end{figure}
\begin{figure}
\centering\includegraphics[width=0.98\columnwidth]{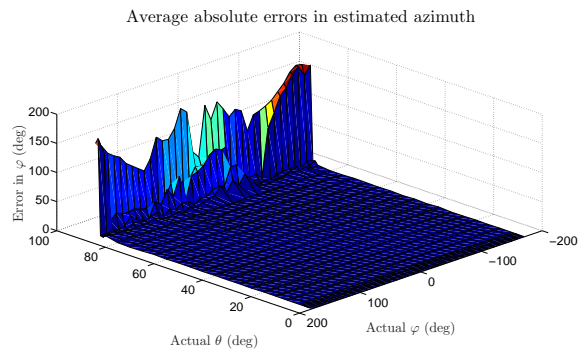}\caption{Average of absolute errors in azimuth estimation using the 3D localization
model. Relatively large errors illustrate the non-observability condition
in azimuth angle estimation with sound source placed around $90^{o}$
elevation, as described by Theorem~\ref{thm:3D} .\label{fig:AzimuthError}}
\end{figure}

To verify the observability conditions for the 3D model as described
by Equations~(\ref{eq:3D state x}) and~(\ref{eq:3d_output}), the
sound source is placed at different locations with a distance of $5\;\text{m}$
from the robot in the simulated room, which evenly cover the hemisphere
above the ground, as shown in Figure\emph{\large{}~}\ref{fig:SoundSourceLocations}.
Figure\emph{\large{}~}\ref{fig:ElevationError} shows the averaged
absolute errors in the elevation estimation versus actual azimuth
and elevation angles of the sound source. Larger errors were observed
when the elevation was close to $0^{o}$, which coincides with Theorem~\ref{thm:3D}.
Figure\emph{\large{}~}\ref{fig:AzimuthError} shows the averaged
absolute errors in the azimuth angle estimation for a single sound
source at different positions. Larger errors were observed when the
elevation was close to $90^{o}$, which again echoes Theorem~\ref{thm:3D}. 

\subsection{Simulation Results for Orientation Localization}

\begin{table*}
\caption{Simulation results of orientation localization for speech \label{tab:Experimental Results Speech}}

\hfill{}%
\begin{tabular}{|c|c|c|c|c|c|c|c|}
\hline 
Expt. & Act. & Act.  & Est. & Avg of abs  & Act. & Est. & Avg of abs \tabularnewline
No. & D(m) & $\varphi$($^{o}$) & $\varphi$($^{o}$) & error ($^{o}$) & $\theta$($^{o}$) & $\theta$($^{o}$) & error ($^{o}$)\tabularnewline
\hline 
\hline 
1 a & 5 & 0 & 0.60 & 0.60 & \multirow{8}{*}{20} & 20.39 & 0.39\tabularnewline
\cline{1-5} \cline{7-8} 
1 b & 5 & 50 & 51.03 & 1.03 &  & 21.44 & 1.44\tabularnewline
\cline{1-5} \cline{7-8} 
1 c & 7 & 90 & 91.21 & 0.21 &  & 20.83 & 0.83\tabularnewline
\cline{1-5} \cline{7-8} 
1 d & 7 & 120 & 121.57 & 1.57 &  & 20.96 & 0.96\tabularnewline
\cline{1-5} \cline{7-8} 
1 e & 3 & 180 & 181.03 & 1.03 &  & 20.16 & 0.16\tabularnewline
\cline{1-5} \cline{7-8} 
1 f & 3 & -40 & -39.33 & 0.67 &  & 19.10 & 0.90\tabularnewline
\cline{1-5} \cline{7-8} 
1 g & 10 & -90 & -88.85 & 1.15 &  & 21.66 & 1.66\tabularnewline
\cline{1-5} \cline{7-8} 
1 h & 10 & -140 & -139.52 & 0.48 &  & 21.18 & 1.18\tabularnewline
\hline 
2 a & 5 & 0 & 2.31 & 2.31 & \multirow{8}{*}{60} & 60.68 & 0.68\tabularnewline
\cline{1-5} \cline{7-8} 
2 b & 5 & 50 & 50.65 & 0.65 &  & 60.53 & 0.53\tabularnewline
\cline{1-5} \cline{7-8} 
2 c & 7 & 90 & 91.79 & 1.79 &  & 60.70 & 0.70\tabularnewline
\cline{1-5} \cline{7-8} 
2 d & 7 & 120 & 121.85 & 1.85 &  & 60.84 & 0.84\tabularnewline
\cline{1-5} \cline{7-8} 
2 e & 3 & 180 & 181.66 & 1.66 &  & 60.05 & 0.05\tabularnewline
\cline{1-5} \cline{7-8} 
2 f & 3 & -40 & -38.66 & 1.34 &  & 60.38 & 0.38\tabularnewline
\cline{1-5} \cline{7-8} 
2 g & 10 & -90 & -89.38 & 0.62 &  & 59.62 & 0.38\tabularnewline
\cline{1-5} \cline{7-8} 
2 h & 10 & -140 & -138.20 & 1.80 &  & 59.78 & 0.22\tabularnewline
\hline 
3 a & 5 & 50 & 50.69 & 0.31 & 0 & 3.39 & 3.39\tabularnewline
\hline 
3 b & 7 & -120 & -119.00 & 1.00 & 4 & 2.40 & 1.60\tabularnewline
\hline 
4 a & 5 & -40 & not def. & not def. & 86 & 90.00 & 4.00\tabularnewline
\hline 
4 b & 7 & 150 & not def. & not def. & 89 & 90.00 & 1.00\tabularnewline
\hline 
\end{tabular}\hfill{}
\end{table*}
\begin{table*}
\caption{Simulation results of orientation localization for white noise \label{tab:Experimental Results Noise}}

\hfill{}%
\begin{tabular}{|c|c|c|c|c|c|c|c|}
\hline 
Expt. & Act. & Act.  & Est. & Avg of abs  & Act. & Est. & Avg of abs \tabularnewline
No. & D(m) & $\varphi$($^{o}$) & $\varphi$($^{o}$) & error ($^{o}$) & $\theta$($^{o}$) & $\theta$($^{o}$) & error ($^{o}$)\tabularnewline
\hline 
\hline 
1 a & 5 & 0 & 1.18 & 1.18 & \multirow{8}{*}{20} & 19.66 & 0.34\tabularnewline
\cline{1-5} \cline{7-8} 
1 b & 5 & 50 & 51.03 & 1.03 &  & 20.44 & 0.44\tabularnewline
\cline{1-5} \cline{7-8} 
1 c & 7 & 90 & 90.25 & 0.25 &  & 20.11 & 0.11\tabularnewline
\cline{1-5} \cline{7-8} 
1 d & 7 & 120 & 121.35 & 1.35 &  & 19.70 & 0.30\tabularnewline
\cline{1-5} \cline{7-8} 
1 e & 3 & 180 & 180.41 & 0.41 &  & 20.48 & 0.48\tabularnewline
\cline{1-5} \cline{7-8} 
1 f & 3 & -40 & -39.44 & 0.56 &  & 19.75 & 0.25\tabularnewline
\cline{1-5} \cline{7-8} 
1 g & 10 & -90 & -89.11 & 0.89 &  & 19.71 & 0.29\tabularnewline
\cline{1-5} \cline{7-8} 
1 h & 10 & -140 & -139.67 & 0.33 &  & 21.18 & 1.18\tabularnewline
\hline 
2 a & 5 & 0 & 1.31 & 1.31 & \multirow{8}{*}{60} & 60.38 & 0.38\tabularnewline
\cline{1-5} \cline{7-8} 
2 b & 5 & 50 & 51.59 & 1.59 &  & 60.39 & 0.39\tabularnewline
\cline{1-5} \cline{7-8} 
2 c & 7 & 90 & 90.74 & 0.74 &  & 60.87 & 0.87\tabularnewline
\cline{1-5} \cline{7-8} 
2 d & 7 & 120 & 121.21 & 1.21 &  & 60.39 & 0.39\tabularnewline
\cline{1-5} \cline{7-8} 
2 e & 3 & 180 & 181.16 & 1.16 &  & 60.51 & 0.51\tabularnewline
\cline{1-5} \cline{7-8} 
2 f & 3 & -40 & -38.66 & 1.34 &  & 60.41 & 0.41\tabularnewline
\cline{1-5} \cline{7-8} 
2 g & 10 & -90 & -88.90 & 1.10 &  & 60.70 & 0.70\tabularnewline
\cline{1-5} \cline{7-8} 
2 h & 10 & -140 & -138.64 & 1.36 &  & 60.57 & 0.57\tabularnewline
\hline 
3 a & 5 & 50 & 51.45 & 1.45 & 0 & 1.57 & 1.57\tabularnewline
\hline 
3 b & 7 & -120 & -118.36 & 1.64 & 4 & 1.57 & 2.43\tabularnewline
\hline 
4 a & 5 & -40 & not def. & not def. & 86 & 90.00 & 4.00\tabularnewline
\hline 
4 b & 7 & 150 & not def. & not def. & 89 & 90.00 & 1.00\tabularnewline
\hline 
\end{tabular}\hfill{}
\end{table*}
A number of experiments were performed to validate the performance
of the proposed SSL technique for orientation localization, as described
in Algorithm~\ref{alg:Algorithm combined model}. White noise and
speech signals were used as a sound source which was placed individually
at different locations in the simulated room with specifications summarized
in Table~\ref{tab:Room specifications-1}. The microphone array was
rotated with an angular velocity of $\omega=2\pi/5$~rad/sec in the
clockwise direction for three complete revolutions. The ITD was calculated
after every $1^{o}$ rotation followed by the estimation performed
using the EKF with parameters given in Table~\ref{tab:EKFparameters}.
Four different sets of experiments were performed keeping the source
at different locations. In first two sets of experiments, the source
was placed in all four quadrants including the axes at different distances,
keeping the elevation constant at $20^{o}$ and $60^{o}$. To validate
the performance of the proposed solution to the non-observability
conditions, other two sets experiments were performed by keeping the
sound source at elevation close to $0^{o}$ and $90^{o}$. The results
of the localization are presented in Tables\emph{\large{}~}\ref{tab:Experimental Results Speech}
and \ref{tab:Experimental Results Noise}. It can be seen that orientation
localization is achieved with errors less than $4^{o}$ using speech
as well as white noise sound source. Large errors are observed when
the elevation of the sound source is around $0^{o}$ and $90^{o}$.
Further, the errors with source elevation around $0^{o}$ is less
as compared to source elevation around $90^{o}$. This was achieved
by using polynomial curve fitting approach mentioned in Section~\ref{subsec:Theta Zero Deg Identification},
with $RMSE_{threshold}=1.9^{o}$, which corresponds to $\theta=15^{o}$
on the fitted curve shown in Figure~\ref{fig:Polyfit} . The value
$d_{threshold}$ was calculated as $0.017$ m (which corresponds to
$\theta=85^{o}$, thereby giving an accuracy of $5^{o}$ when the
sound source gets close to $90^{o}$ elevation) for the simulated
environment with specification given in Table~\ref{tab:Room specifications-1}. 

\subsection{Simulation Results for Distance Localization}

\begin{table}
\caption{Simulation results of distance localization using speech sound source
\label{tab:Experimental Results Distance Speech}}

\hfill{}%
\begin{tabular}{|c|c|c|c|c|c|}
\hline 
Expt. & Act. & Act. & Act. & Est. & Avg of abs \tabularnewline
No. & $\varphi$($^{o}$) & $\theta$($^{o}$) & D(m) & D(m) & error (m)\tabularnewline
\hline 
\hline 
1 a & 0 & \multirow{8}{*}{20} & 5 & 5.01 & 0.01\tabularnewline
\cline{1-2} \cline{4-6} 
1 b & 50 &  & 5 & 5.01 & 0.01\tabularnewline
\cline{1-2} \cline{4-6} 
1 c & 90 &  & 7 & 6.94 & 0.06\tabularnewline
\cline{1-2} \cline{4-6} 
1 d & 120 &  & 7 & 6.93 & 0.07\tabularnewline
\cline{1-2} \cline{4-6} 
1 e & 180 &  & 3 & 3.01 & 0.01\tabularnewline
\cline{1-2} \cline{4-6} 
1 f & -40 &  & 3 & 3.01 & 0.01\tabularnewline
\cline{1-2} \cline{4-6} 
1 g & -90 &  & 10 & 9.54 & 0.46\tabularnewline
\cline{1-2} \cline{4-6} 
1 h & -140 &  & 10 & 9.81 & 0.19\tabularnewline
\hline 
2 a & 0 & \multirow{8}{*}{60} & 5 & 5.02 & 0.02\tabularnewline
\cline{1-2} \cline{4-6} 
2 b & 50 &  & 5 & 5.02 & 0.02\tabularnewline
\cline{1-2} \cline{4-6} 
2 c & 90 &  & 7 & 6.94 & 0.06\tabularnewline
\cline{1-2} \cline{4-6} 
2 d & 120 &  & 7 & 6.94 & 0.06\tabularnewline
\cline{1-2} \cline{4-6} 
2 e & 180 &  & 3 & 3.00 & 0.00\tabularnewline
\cline{1-2} \cline{4-6} 
2 f & -40 &  & 3 & 3.01 & 0.01\tabularnewline
\cline{1-2} \cline{4-6} 
2 g & -90 &  & 10 & 9.52 & 0.48\tabularnewline
\cline{1-2} \cline{4-6} 
2 h & -140 &  & 10 & 9.41 & 0.59\tabularnewline
\hline 
3 a & 50 & 0 & 5 & 5.02 & 0.02\tabularnewline
\hline 
3 b & -120 & 4 & 7 & 6.87 & 0.13\tabularnewline
\hline 
4 a & -40 & 86 & 5 & 5.02 & 0.02\tabularnewline
\hline 
4 b & 150 & 89 & 7 & 6.83 & 0.17\tabularnewline
\hline 
\end{tabular}\hfill{}
\end{table}
\begin{table}[h]
\caption{Simulation results of distance localization using white noise sound
source \label{tab:Experimental Results Distance Noise}}

\hfill{}%
\begin{tabular}{|c|c|c|c|c|c|}
\hline 
Expt. & Act. & Act. & Act. & Est. & Avg of abs \tabularnewline
No. & $\varphi$($^{o}$) & $\theta$($^{o}$) & D(m) & D(m) & error (m)\tabularnewline
\hline 
\hline 
1 a & 0 & \multirow{8}{*}{20} & 5 & 5.01 & 0.01\tabularnewline
\cline{1-2} \cline{4-6} 
1 b & 50 &  & 5 & 5.01 & 0.01\tabularnewline
\cline{1-2} \cline{4-6} 
1 c & 90 &  & 7 & 6.92 & 0.08\tabularnewline
\cline{1-2} \cline{4-6} 
1 d & 120 &  & 7 & 6.92 & 0.08\tabularnewline
\cline{1-2} \cline{4-6} 
1 e & 180 &  & 3 & 3.01 & 0.01\tabularnewline
\cline{1-2} \cline{4-6} 
1 f & -40 &  & 3 & 3.01 & 0.01\tabularnewline
\cline{1-2} \cline{4-6} 
1 g & -90 &  & 10 & 9.52 & 0.48\tabularnewline
\cline{1-2} \cline{4-6} 
1 h & -140 &  & 10 & 9.44 & 0.56\tabularnewline
\hline 
2 a & 0 & \multirow{8}{*}{60} & 5 & 5.01 & 0.01\tabularnewline
\cline{1-2} \cline{4-6} 
2 b & 50 &  & 5 & 5.01 & 0.01\tabularnewline
\cline{1-2} \cline{4-6} 
2 c & 90 &  & 7 & 6.92 & 0.08\tabularnewline
\cline{1-2} \cline{4-6} 
2 d & 120 &  & 7 & 6.92 & 0.08\tabularnewline
\cline{1-2} \cline{4-6} 
2 e & 180 &  & 3 & 3.01 & 0.01\tabularnewline
\cline{1-2} \cline{4-6} 
2 f & -40 &  & 3 & 3.01 & 0.01\tabularnewline
\cline{1-2} \cline{4-6} 
2 g & -90 &  & 10 & 9.48 & 0.52\tabularnewline
\cline{1-2} \cline{4-6} 
2 h & -140 &  & 10 & 9.43 & 0.57\tabularnewline
\hline 
3 a & 50 & 0 & 5 & 5.01 & 0.01\tabularnewline
\hline 
3 b & -120 & 4 & 7 & 6.89 & 0.11\tabularnewline
\hline 
4 a & -40 & 86 & 5 & 5.01 & 0.01\tabularnewline
\hline 
4 b & 150 & 89 & 7 & 6.90 & 0.10\tabularnewline
\hline 
\end{tabular}\hfill{}
\end{table}
\begin{figure}
\includegraphics[clip,width=1.1\columnwidth]{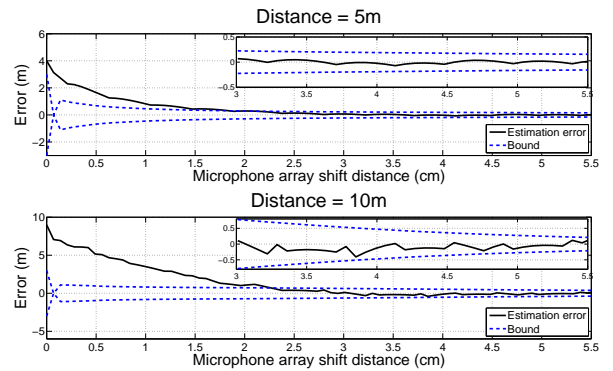}

\caption{Simulation results for distance estimation using EKF. A single sound
source was placed at two different locations with distances of $5$
m and $10$ m, respectively. The bounds represent the three standard
deviation of the estimation error.\label{fig:Distance_subplot}}
\end{figure}

Speech and white-noise sounds were also used to test the performance
of the distance localization. A single sound source was placed at
different locations and the ITD signal was recorded while the microphone
array was continuously shifted for $200$ steps each with a distance
of $\triangle d=0.0007$ m. The results are summarized in Tables~\ref{tab:Experimental Results Distance Speech}
and \ref{tab:Experimental Results Distance Noise}. The key parameters
of the EKF are given in Table~\ref{tab:EKFparameters}. The results
for the distance localization with a sound source placed at different
locations are shown in Figure\emph{\large{}~}\ref{fig:Distance_subplot}.
It is observed that the error in the estimation converges quickly
and a total shift of microphone array of approximately $3\text{ \text{cm}}$
is sufficient for the estimates to completely converge to and remain
in the three standard deviation bounds. The average of absoute error
in the estimation is found to be less than $0.6$ m in both the case
of speech as well as white noise sound sources. 

\section{Experimental Results\label{sec:Experimental-Results}}

Experiments were conducted using two different hardware platforms:
a KEMAR dummy head in a well equipped hearing laboratory and a robotic
platform equipped with a set of two rotational microphones. The following
subsections discuss the hardware platforms and the results.

\subsection{Results using KEMAR Dummy Head}

Experiments using the KEMAR dummy head were conducted in a high frequency
focused sound treated room~\cite{yost2014sound} with dimension $4.6\text{ }\text{m}$
x $3.7\text{ }\text{m}$ x $2.7\text{ }\text{m}$ as shown in Figure~\ref{fig:Setup-of-KEMAR}.
The ITD however is mostly effective for low frequency sounds below
1.5 kHz as a spatial hearing cue~\cite{middlebrooks1991sound}. The
walls, floor, and ceiling of the room were covered by polyurethane
acoustic foam with a thickness of only 5 cm which is relatively low
compared to the sound wavelength thereby making a relatively low reduction
in low and middle frequencies~\cite{beranek2012acoustics}, thereby
making it a challenging acoustic environment. For broad band noise,
T60 (i.e., the time required for the sound level to decay 60 dB~\cite{Sabine1922})
was 97 ms. In an octave band centered at 1000 Hz, T60 for the noise
was on an average of 324 ms. 
\begin{figure}[h]

\centering\includegraphics[height=0.7\columnwidth]{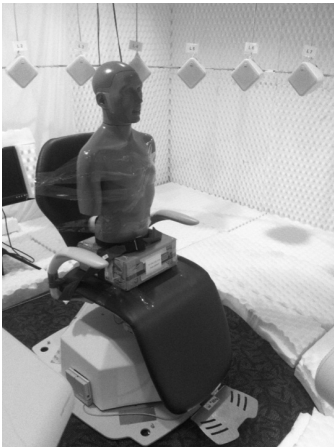}\caption{Setup of the KEMAR dummy head on a rotating chair in the middle of
the sound treated room~\cite{Sun2015}. \label{fig:Setup-of-KEMAR}}

\end{figure}

The digitally generated audio signals using a MATLAB program and three
12-channel Digital-to-Analog converters running at 44,100 cycles each
second per channel were amplified using AudioSource AMP 1200 amplifiers
before they were played from an array of 36 loudspeakers. The two
microphones were installed on the KEMAR dummy head temporarily mounted
on a rotating chair which was rotated at an approximate rate of 32\textdegree /s
for about two circles in the middle of the room. The data collected
in the second rotation was used for the EKF. Motion data was collected
by the gyroscope mounted on the top of the dummy head. The audio signals
were amplified and collected by a sound card which were then stored
on a desktop computer for further processing. The ITD was processed
with a generalized cross-correlation model~\cite{Knapp1976The} in
each time frame corresponding to the 120 Hz sampling rate of the gyroscope.
The computation was completed by a MATLAB program on a desktop computer.
Raw data with a single sound source located at four different locations
were collected.
\begin{figure}[h]
\includegraphics[bb=0cm 0cm 666bp 683bp,width=0.55\columnwidth]{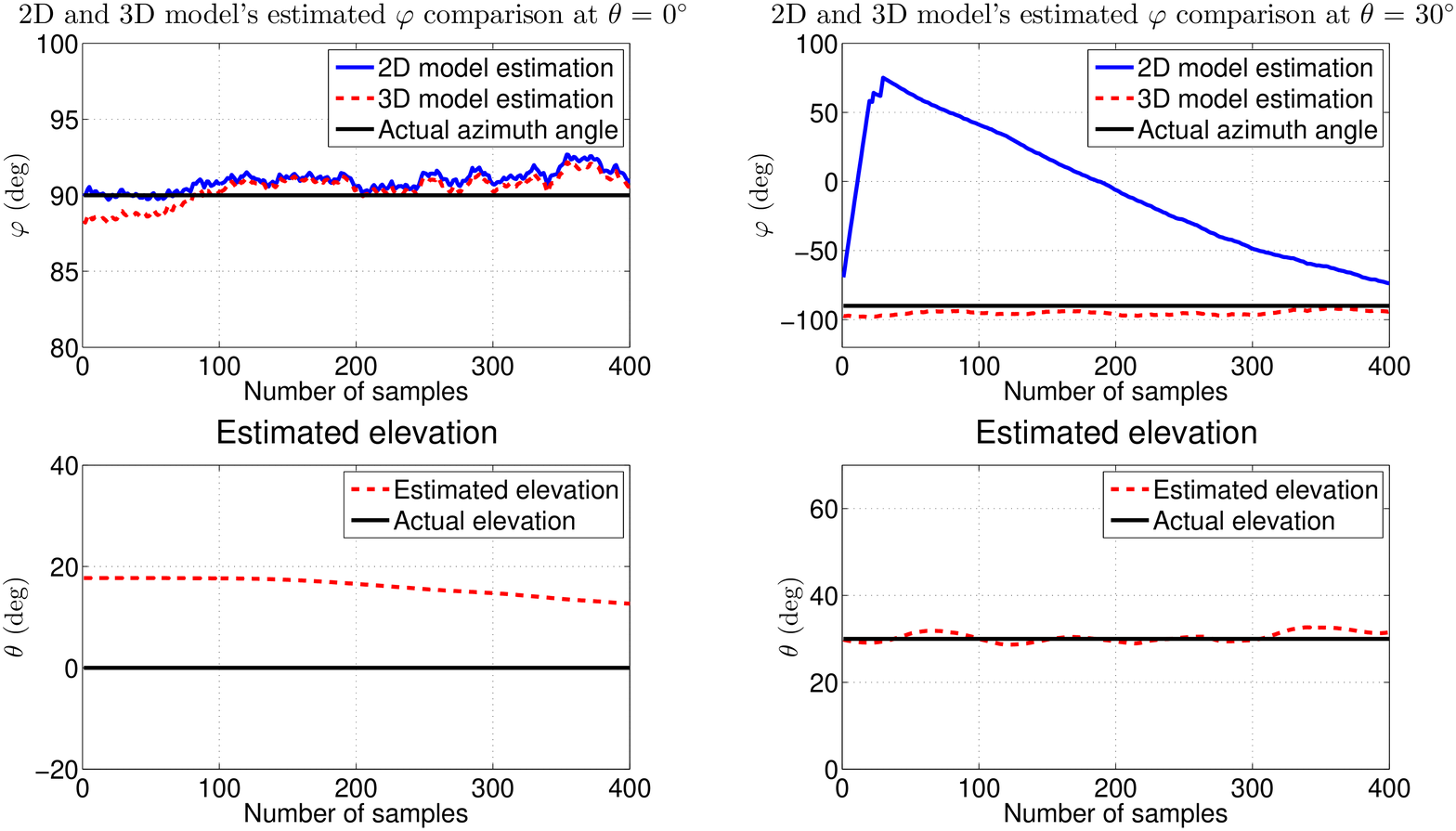}

\caption{Experimental results for orientation localization using the KEMAR
dummy head. When $\theta=0^{o}$, the azimuth estimates using the
2D and 3D models are very close (in the top-left figure), which implies
the elevation estimates are not reliable (in the bottom-left figure).
When $\theta=30^{o}$, the azimuth estimates are obviously different
(in the top-right figure), which implies reliable elevation estimates
using the 3D model (in the bottom-right figure). \label{fig: RawData subplot}}
\end{figure}

The left two subfigures in Figure~\ref{fig: RawData subplot} are
generated when the actual elevation angle is $0^{o}$. It can be seen
that the azimuth estimations using the 2D and 3D models  are very
close, which implies that the actual elevation angle is close to $0^{o}$
and the elevation estimation using the 3D model is not reliable. The
right two subfigures in Figure~\ref{fig: RawData subplot} are generated
when the actual elevation angle is $30^{o}$. It can be seen that
the azimuth estimations using the 2D and 3D localization models are
obviously different while the elevation estimation using the 3D model
is fairly accurate, which verifies the proposed algorithm shown in
Figure~\ref{fig:Flowchart combined model}. Table~\ref{tab:Raw Data results-1}
shows the estimation results obtained using the 3D localization model.
It can be seen that the RMSE of the difference between the estimated
azimuth values using respectively the 2D and 3D models works well
in checking the zero elevation condition. 
\begin{table*}[h]
\caption{Experimental results using KEMAR dummy head: Orientation localization
using the 3D model. (RMSE: difference between azimuth estimations
using the 2D and 3D models, respectively) \label{tab:Raw Data results-1}}

\hfill{}%
\begin{tabular}{|c|c|c|c|c|c|c|c|}
\hline 
Expt. & Act.  & Est. & Avg of abs  & RMSE & Act. & Est. & Avg of abs \tabularnewline
No. & $\varphi$($^{o}$) & $\varphi$($^{o}$) & error ($^{o}$) & ($^{o}$) & $\theta$($^{o}$) & $\theta$($^{o}$) & error ($^{o}$)\tabularnewline
\hline 
\hline 
1 & 90 & 91.21 & 1.21 & 1.39 & \multicolumn{1}{c|}{0} & 13.64 & 13.64\tabularnewline
\hline 
2 & -20 & -21.53 & 1.53 & 1.16 & 0 & 48.14 & 48.14\tabularnewline
\hline 
3 & 90 & 90.40 & 0.40 & 79.94 & 60 & 59.05 & 0.95\tabularnewline
\hline 
\end{tabular}\hfill{}
\end{table*}

\subsection{Results using Robotic Platform}

Experiments were also performed using a robotic platform shown in
Figure~\ref{fig:Binaural-robot.}. In these experiments, two microelectromechanical
systems (MEMS) analog/digital microphones were used for recording
the sound signal coming from the sound source. Flex adapters were
used to hold the microphones. The angular speed of the rotation of
the microphone array was controlled by a bipolar stepper motor with
gear ratio adjusted to $0.9^{o}$ per step. The stepper motor was
controlled by an Arduino microprocessor. The distance between two
microphones was kept constant as $0.3\text{ }\text{m}$. An audio
(music) was played in a loud speaker which was used as a sound source
kept at different locations. The estimation results are shown in Figure~\ref{fig:OurRobot_subplot}
and Table~\ref{tab:Our robot results-1}.

\begin{figure}[h]
\centering\includegraphics[width=0.7\columnwidth]{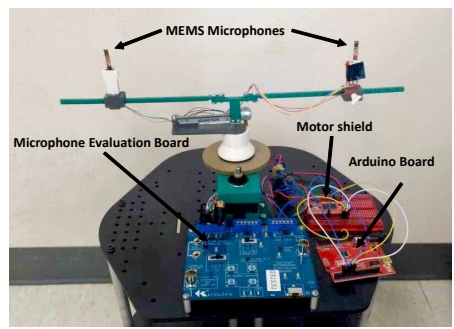}\caption{A two-microphone system is equipped on a ground robot.\label{fig:Binaural-robot.}}
\end{figure}
\begin{figure}[h]
\includegraphics[bb=0bp 0bp 936bp 746bp,width=0.8\columnwidth]{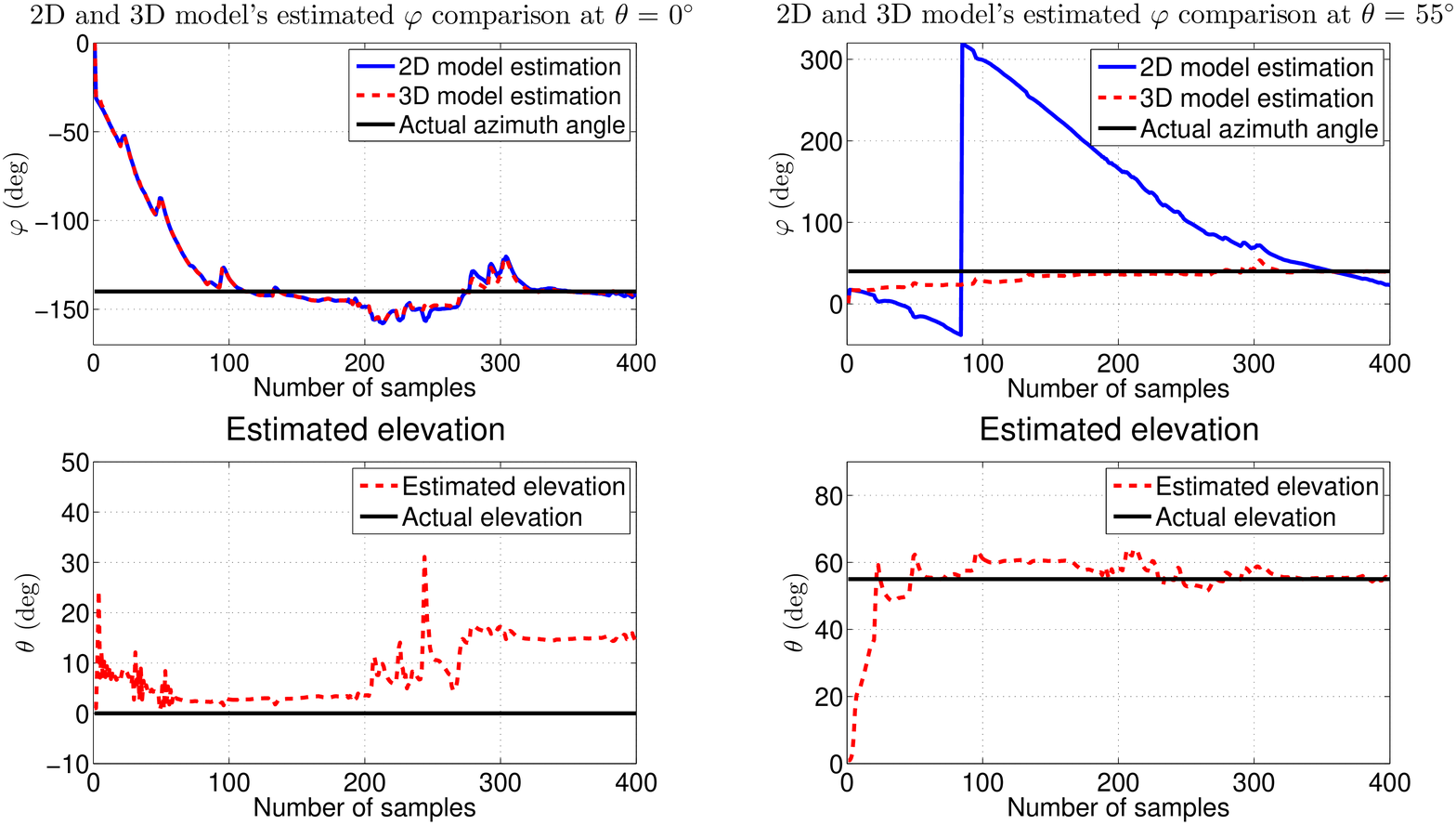}

\caption{Experimental results for orientation localization using the robotic
platform. When $\theta=0^{o}$, the azimuth estimates using the 2D
and 3D models are very close (in the top-left figure), which implies
the elevation estimates are not reliable (in the bottom-left figure).
When $\theta=55^{o}$, the azimuth estimates are obviously different
(in the top-right figure), which implies reliable elevation estimates
using the 3D model (in the bottom-right figure).\label{fig:OurRobot_subplot}}
\end{figure}
\begin{table*}[h]
\caption{Experimental results using the robotic platform: Orientation localization
using 3D model (RMSE: difference between azimuth estimations using
the 2D and 3D models, respectively) \label{tab:Our robot results-1}}

\hfill{}%
\begin{tabular}{|c|c|c|c|c|c|c|c|}
\hline 
Expt. & Act.  & Est. & Avg of abs  & RMSE & Act. & Est. & Avg of abs \tabularnewline
No. & $\varphi$($^{o}$) & $\varphi$($^{o}$) & error ($^{o}$) & ($^{o}$) & $\theta$($^{o}$) & $\theta$($^{o}$) & error ($^{o}$)\tabularnewline
\hline 
\hline 
1  & -140 & -140.65 & 0.65 & 0.72 & \multicolumn{1}{c|}{0} & 14.96 & 14.96\tabularnewline
\hline 
2 & 180 & 178.71 & 1.29 & 0.69 & 5 & 11.59 & 6.59\tabularnewline
\hline 
3 & 40 & 39.67 & 0.33 & 8.80 & 55 & 55.24 & 0.24\tabularnewline
\hline 
4 & 40 & 38.20 & 1.80 & 10.96 & 65 & 64.67 & 0.33\tabularnewline
\hline 
\end{tabular}\hfill{}
\end{table*}

It can be seen that the azimuth estimations using the 2D and 3D models
shown in the top-left subfigure in Figure~\ref{fig:OurRobot_subplot}
generated when the actual elevation angle is $0^{o}$ are very close,
which implies that the elevation is close to $0^{o}$ and the elevation
estimation shown in the bottom-left subfigure in Figure~\ref{fig:OurRobot_subplot}
using the 3D localization model is not reliable. However, the two
subfigures on the right in Figure~\ref{fig:OurRobot_subplot} are
generated by keeping the sound source at an elevation angle of $55^{o}$.
As proposed in the algorithm shown in Figure~\ref{fig:Flowchart combined model},
the azimuth estimations using the 2D and 3D localization models are
different while the elevation estimation using the 3D model is fairly
accurate. Table~\ref{tab:Our robot results-1} shows the estimation
results obtained using the 3D localization model. It can be seen that
the zero elevation condition can be checked using the RMSE of the
difference between the estimated azimuth values using respectively
the 2D and 3D models.

A fitted curve similar to one shown in the Figure~\ref{fig:Polyfit}
can be generated for the environment by keeping the sound source at
different elevation angles and recording the $RMSE$ values between
$\varphi_{2D}$ and $\varphi_{3D}$ estimations. The value of the
parameter $RMSE_{threshold}$ can be decided, which can be used to
check the $\theta=0^{o}$ scenario. Further, the generated fitted
curve can be used to give a closer estimation of the elevation angle.

\section{Conclusion\label{sec:Conclusion-and-Future}}

This paper presents a novel technique that performs a complete localization
(i.e., both orientation and distance) of a stationary sound source
in a three-dimensional (3D) space. Two singular conditions when unreliable
orientation localization (the elevation angle equals $0$ or $90^{o}$)
occurs were found by using the observability theory. The root-mean-squared
error (RMSE) value of the difference between the azimuth estimates
using respectively the 2D and 3D models was used to check the $0^{o}$
elevation condition and the elevation was further estimated using
a polynomial curve fitting technique. The $90^{o}$ elevation was
detected by checking zero-ITD signal. Based on an accurate orientation
localization, the distance localization was done by first rotating
the microphone array to face toward the sound source and then shifting
the microphones perpendicular to the source-robot vector by a distance
of a fixed number of steps. Under challenging acoustic environments
with relatively low-energy targets and high-energy noise, high localization
accuracy was achieved in both simulations and experiments. The mean
of the average of absolute estimation error was less than $4^{o}$
for angular localization and less than $0.6$ m for distance localization
in simulation results and techniques to detect $\theta=0^{o}$ and
$90^{o}$ are verified in both simulation and experimental results. 
\begin{acknowledgements}
Acknowledgment
\end{acknowledgements}

The authors would like to thank Dr. Xuan Zhong for providing with
the experimental raw data using the KEMAR dummy head.

\bibliographystyle{../SpringerTemplate/spmpsci}
\bibliography{References}

\end{document}